\pdfoutput=1
\PassOptionsToPackage{table}{xcolor}
\PassOptionsToPackage{framemethod=tikz}{mdframed}
\documentclass[11pt,a4paper]{gdm/googledeepmind}
\usepackage{mathtools}
\usepackage{bm}
\usepackage{multirow}
\usepackage[section]{placeins}
\usepackage{algorithm}
\usepackage{algpseudocode}
\usepackage{natbib}
\usepackage[most]{tcolorbox}
\usepackage{thmtools}
\usepackage{wrapfig}
\usepackage{gdm/gdm-colors}
\graphicspath{{img/}}

\definecolor{impr}{RGB}{34,139,34}
\definecolor{lightred}{RGB}{255,230,230}
\definecolor{darkred}{RGB}{192,0,0}
\definecolor{bestbg}{HTML}{FFF2B2}
\definecolor{secondbg}{HTML}{E8F0FE}
\definecolor{LavenderLight}{HTML}{C7C3F5}
\definecolor{LightCoral}{RGB}{240,128,128}
\definecolor{LightBlue}{RGB}{173,216,230}
\newcommand{\best}[1]{\cellcolor{bestbg}{\boldmath\textbf{#1}}}
\newcommand{\sbest}[1]{\cellcolor{secondbg}\underline{#1}}

\renewcommand{\arraystretch}{1.15}

\colorlet{thmboxbg}{LightBlue!20}
\colorlet{remarkboxbg}{LightCoral!10}
\mdfdefinestyle{thmblue}{%
    backgroundcolor=thmboxbg, linecolor=black, linewidth=0.5pt, roundcorner=3pt,
    innertopmargin=6pt, innerbottommargin=7pt, innerleftmargin=8pt,
    innerrightmargin=8pt, skipabove=8pt, skipbelow=8pt, nobreak=true}
\mdfdefinestyle{thmcoral}{%
    backgroundcolor=remarkboxbg, linecolor=black, linewidth=0.5pt, roundcorner=3pt,
    innertopmargin=6pt, innerbottommargin=7pt, innerleftmargin=8pt,
    innerrightmargin=8pt, skipabove=8pt, skipbelow=8pt, nobreak=true}

\declaretheorem[name=Theorem, numberwithin=section, mdframed={style=thmblue}]{theorem}
\declaretheorem[name=Lemma, sibling=theorem, mdframed={style=thmblue}]{lemma}
\declaretheorem[name=Proposition, sibling=theorem, mdframed={style=thmblue}]{proposition}
\declaretheorem[name=Corollary, sibling=theorem, mdframed={style=thmblue}]{corollary}
\declaretheorem[name=Remark, sibling=theorem, style=remark, mdframed={style=thmcoral}]{remark}
\declaretheorem[name=Definition, sibling=theorem, style=definition]{definition}
\declaretheorem[name=Assumption, sibling=theorem, style=definition]{assumption}
\declaretheorem[name=Takeaway, numbered=no, style=remark, mdframed={style=thmcoral}]{takeaway}

\newcommand{\method}{\textsc{Trace}}
\newcommand{\drbg}{\mathrm{DRBG}}
\newcommand{\ctx}{\mathrm{ctx}}
\newcommand{\keyone}{\mathsf{key}_1}
\newcommand{\keytwo}{\mathsf{key}_2}
\newcommand{\Hnull}{H_0}
\newcommand{\Halt}{H_1}
\newcommand{\E}{\mathbb{E}}
\newcommand{\Var}{\mathrm{Var}}
\newcommand{\ent}{\mathcal{H}}
\DeclareMathOperator*{\argmin}{arg\,min}
\DeclareMathOperator*{\argmax}{arg\,max}

\usepackage{longtable}
\usetikzlibrary{positioning,arrows.meta}

\title{\method: A Two-Channel Robust Attribution Watermark via Complementary Embeddings for LLM-Agent Trajectories}
\author{%
\begin{minipage}{\linewidth}
\centering\normalfont
\makebox[0.31\linewidth]{Zheng Gao$^{1}$}%
\makebox[0.31\linewidth]{Xiaoyu Li$^{1}$}%
\makebox[0.31\linewidth]{Xiaoyan Feng$^{2}$}\par\vspace{3pt}
\makebox[0.31\linewidth]{Jiaojiao Jiang$^{1}$}%
\makebox[0.31\linewidth]{Yang Song$^{1}$}%
\makebox[0.31\linewidth]{Yulei Sui$^{1}$}\par\vspace{3pt}
\makebox[0.31\linewidth]{Zhenchang Xing$^{3}$}%
\makebox[0.31\linewidth]{Liming Zhu$^{3}$}\par\vspace{8pt}
{\small
$^{1}$University of New South Wales \quad
\texttt{\{zheng.gao1,\,xiaoyu.li2,\,jiaojiao.jiang,\,yang.song1,\,y.sui\}@unsw.edu.au}\par\vspace{2pt}
$^{2}$Griffith University \quad
\texttt{xiaoyan.feng@griffithuni.edu.au}\par\vspace{2pt}
$^{3}$CSIRO's Data61 \quad
\texttt{\{zhenchang.xing,\,liming.zhu\}@data61.csiro.au}\par}
\end{minipage}%
}

\correspondingauthor{}

\begin{abstract}
LLM agents reach users through resellers, who may rebrand a developer's agent or substitute a cheaper model. When provenance is disputed, attribution rests on the trajectory log (the record of tool calls, observations, and executed actions, not the model's reasoning), which the reseller stores and processes to meter usage. A watermark must therefore survive an adversary with full read/write access to the very evidence it is detected from; existing agent watermarks do not, as their attribution is read straight off that log. We present TRACE, to our knowledge the first agent watermark that is distortion-free in its action choices, self-synchronizing under deletion, and unconditionally invariant under rewriting. Deletion desynchronizes a position-derived key and rewriting alters content, so a deletion-robust key must come from content and a rewrite-robust key from position, and no single key serves both. A trajectory, however, has room for two watermarks. TRACE superposes a selection channel that sets which action is chosen, keyed on local content with a distortion-free sampler, so the agent's distribution is provably unchanged and detection resynchronizes after deletions, and a tally channel that sets how many records each decision group holds, keyed on the log's skeleton alone, which no rewriting can touch. We prove this behavioral watermark's signal is bought with decision entropy, each decision paying at least half its entropy and deterministic decisions nothing, and that erasing both channels forces the reseller to corrupt the trajectories it resells. On ToolBench and ALFWorld, TRACE matches the unwatermarked agent's success rate while its selection channel reaches detection scores near z = 100 on long-horizon trajectories, stays detectable under 70\% step deletion, and keeps a tally channel exactly unchanged under LLM rewriting of any strength.
\end{abstract}

\begin{document}

\maketitle

\addtocontents{toc}{\protect\setcounter{tocdepth}{-3}}

\begin{center}
\vspace{-6pt}
\includegraphics[width=0.56\textwidth]{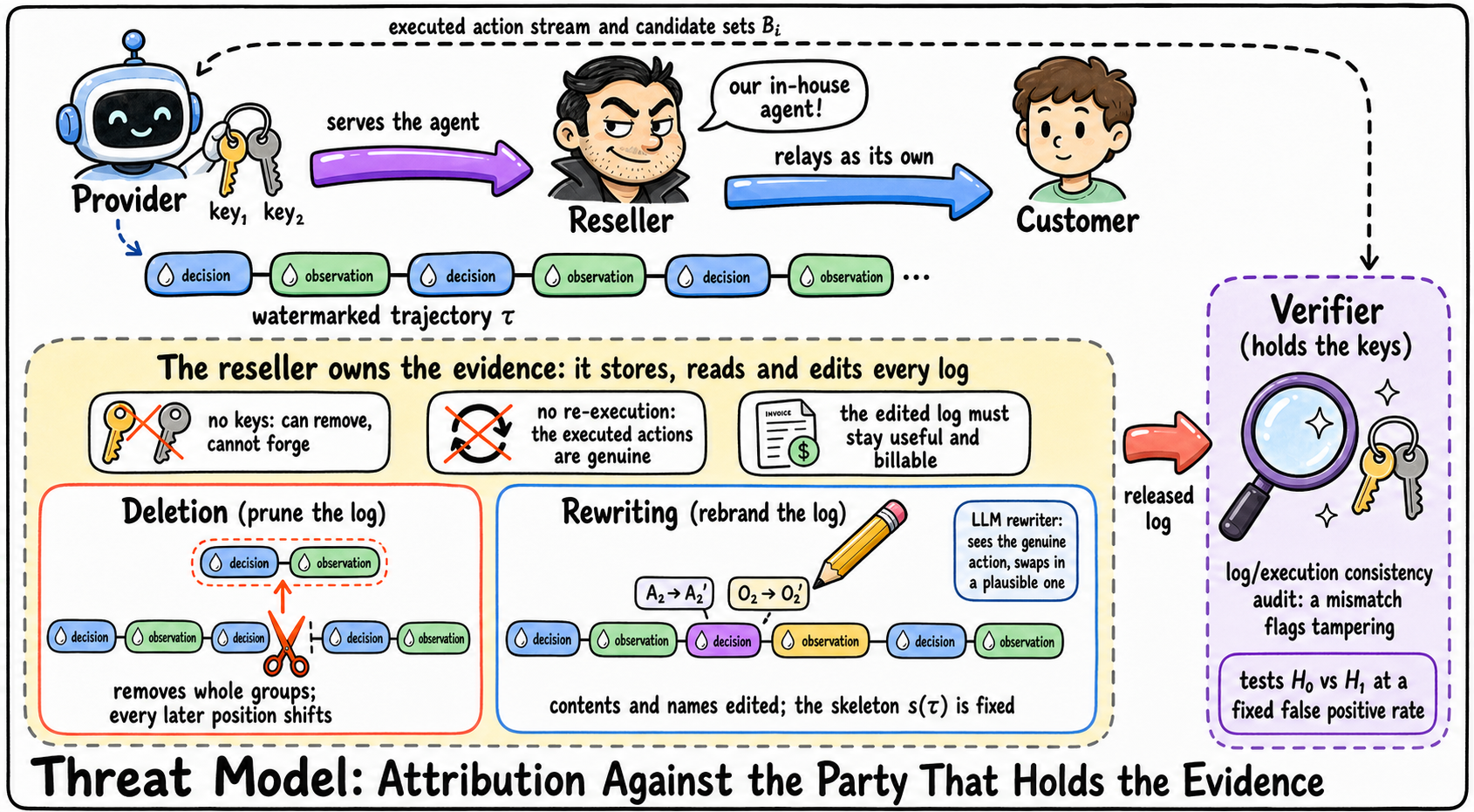}
\captionof{figure}{The reseller threat model: the adversary owns the evidence it audits.}
\label{fig:threat}
\vspace{-6pt}
\end{center}

\definecolor{promptframe}{RGB}{120,113,108}
\definecolor{promptback}{RGB}{250,250,248}
\newmdenv[
  linecolor=promptframe,
  backgroundcolor=promptback,
  linewidth=0.9pt,
  roundcorner=3pt,
  nobreak=true,
  innertopmargin=7pt, innerbottommargin=7pt,
  innerleftmargin=9pt, innerrightmargin=9pt,
  skipabove=9pt, skipbelow=9pt,
  frametitlefont={\normalfont\small\bfseries},
  frametitlerule=true, frametitlerulewidth=0.4pt,
]{promptbox}

\section{Introduction}\label{sec:intro}

Large language model agents no longer merely produce text: they invoke
search APIs, file tickets, send messages, book services, execute
code, and respond to security incidents
\citep{yao2022react,schick2023toolformer,qin2024toolllm,Park2023GenerativeAI,li2026cyberthreatsurvey}.
Actions carry consequences that prose does not. An operator audited
after an incident must show which of the logged actions its agent did
and did not take. When agent behavior causes harm, attribution is the
first step of liability. Governance proposals reach the same point
from the policy side, calling for visibility into agent activity
through identifiers and activity logs \citep{Chan2024VisibilityIA}. Every
one of these needs runs through the same artifact, the agent's
trajectory log, and the log serves them only if it can be attributed
to the agent that produced it. Throughout, \emph{trajectory} means this
execution trace, the logged tool calls, observations, and actions,
not the model's reasoning trace or a bare conversation history. A
system exposing no such trace (a chat model's single response, an
image generator's single image) presents no decision sequence for
\method{} to mark.

For text, provenance has a mature answer: watermarking. Biasing or
derandomizing the token sampler with a secret key lets a detector that
holds the key distinguish watermarked output from natural text
\citep{kirchenbauer2023watermark,aaronson2023watermarking,kuditipudi2023robust,christ2024undetectable},
and the approach is deployed at production scale
\citep{dathathri2024scalable}. For agents, however, token-level
watermarks are the wrong primitive, for three reasons. First, the log
does not store the token stream. An agent's decisions are translated
into structured records of behaviors and actions, tool calls, their
arguments, and the environment's observations, and the sampled tokens
that would carry a text watermark are largely lost in that translation
\citep{huang2025agent,huang2026agentmark}. The artifact that
survives, and the one that matters for attribution, is the action
stream. Second, the signal-bearing units are scarce. A trajectory
carries its identity in the decision sequence, and decisions number in
the single digits to a few dozen per task: we measure roughly $1.4$
effective decision groups per task on ToolBench and $23$ to $25$
decision steps per task on ALFWorld, far below the token counts at
which text watermarks attain power. Third, agent decisions are frequently
low-entropy, often admitting a single tool, and any watermark that
moves probability mass at such steps directly damages task success.

The watermark has accordingly begun to move from the tokens to the
behavior. Agent Guide \citep{huang2025agent} biases the agent's
high-level choices toward a keyed subset and detects the bias with a
z-statistic. AgentMark \citep{huang2026agentmark} removes the bias,
embedding a multi-bit identifier into planning decisions through
distribution-preserving conditional sampling under black-box APIs.
AgentWM \citep{wang2026protecting} biases selection among semantically
equivalent tool paths, so that the signal survives into models
trained to imitate the agent. ActHook \citep{meng2026watermarking}
watermarks trajectories as training data, planting keyed hook actions
that resurface in any model trained on a stolen log. The decision
stream, these works establish, is a viable carrier.

They differ in goal, from provenance to imitation defense to dataset
copyright, but share one structure: each carries its signal once,
under one keying, and measures robustness empirically against generic
perturbations. Their adversaries perturb logs, paraphrase text, or
train imitators. None of them holds the evidence itself. Yet that is
exactly the party through which agents increasingly reach users. A
\emph{reseller} licenses a developer's agent and forwards it to
customers, but advertises it as something it is not
(Figure~\ref{fig:threat}): as the
reseller's own in-house system, or as a premium, safety-evaluated
provider while a cheaper model runs underneath, a substitution
pattern cost-aware LLM deployments already practice
systematically~\citep{zhang2026aprrouter}. Metering and billing
already require the reseller to process every trajectory, so it is
entitled to store, read, and edit the very logs from which provenance
would be judged. When a harmed customer, a contested contract, or a
platform auditing a partner's traffic raises the question \emph{did
this trajectory come from that agent?}, the answer must be
established against an adversary who owns the evidence. Attribution
that the log-holder can quietly strip is no attribution at all.

The reseller's two natural laundering moves pull keying design in
opposite directions. To shrink and sanitize what it forwards, the
reseller \emph{prunes} the log: it drops steps, truncates, and
discards records that betray the upstream provider. This is a
\emph{deletion} attack. Deletion desynchronizes any position-derived
key, so surviving it demands a key derived from local \emph{content},
which lets detection re-align at the very next decision after a dropped
record. To rebrand what it forwards,
the reseller \emph{rewrites} the log: it paraphrases observations and
renames tools into its own namespace while preserving the step
structure. This is a \emph{rewriting} attack. Rewriting can alter
every content field, so surviving it demands a key derived from
\emph{position}, which content edits cannot move. No single keying
satisfies both demands, which is why a single-signal scheme, whatever
its keying, falls to one move or the other. A trajectory, however, has
room for two watermarks: one on the identity of its decisions, one on
their shape.

\paragraph{Our approach.}
\method{} superposes two watermarks on one trajectory, with
independent keys, disjoint carriers, and, by design, disjoint
vulnerable surfaces. It is, to our knowledge, the first agent
watermark that is distortion-free in its action choices,
self-synchronizing under deletion, and unconditionally invariant
under rewriting. The \emph{selection} channel modulates \emph{which}
action is selected, sampling the agent's distribution through a
distortion-free exponential race \citep{kuditipudi2023robust} keyed on
preceding content. The sampled distribution provably equals the
agent's (Theorem~\ref{thm:df}), and content keying confines a
deletion's damage to one neighboring key, so detection resynchronizes
(Proposition~\ref{prop:blast}). Zero distortion does not make
detection free: each decision's signal is lower-bounded by half its
entropy and vanishes exactly at deterministic decisions, so short
trajectories are pooled, at an explicit rate
(Theorem~\ref{thm:entropy}, Corollary~\ref{cor:power}). The
\emph{tally} channel modulates \emph{how many} records each decision
group contains, appending a context-neutral redundant record under a
key derived from group position alone. Count and key are functions of
the trajectory's skeleton, which rewriting cannot touch, so the
channel is unconditionally invariant under every rewriting attack
(Theorem~\ref{thm:rewrite}).

Erasing both layers at once is provably expensive: the attack must
edit the skeleton and, unless it targets groups by their realized
scores, alter a constant fraction of group contents
(Theorem~\ref{thm:erasure}). Both edits degrade the service the
reseller is paid to deliver. Laundering the log means corrupting the
product.

\paragraph{Contributions.}
We make the following contributions:
\begin{itemize}
\item We formalize attribution against the party that holds the
evidence: a reseller with full read and write access to the trajectory
log, whose two laundering moves induce the deletion and rewriting
attack classes (Section~\ref{sec:prelim}).
\item We design \method{}, which superposes a content-keyed selection
channel and a position-keyed tally channel on one trajectory: to our
knowledge the first agent watermark that is distortion-free,
deletion-self-synchronizing, and rewrite-invariant at once
(Section~\ref{sec:method}).
\item We back the design with guarantees proved structurally where
prior agent watermarks measure robustness empirically: exact null laws
for both detectors, an entropy lower bound that prices distortion-free
detectability and yields the pooling rate, and a joint-erasure theorem
charging any attack that silences both channels with skeleton edits
and, when obliviously targeted, a constant fraction of altered groups
(Section~\ref{sec:theory}).
\item We introduce the \emph{LLM rewriter}, to our knowledge the first
informed, plausibility-preserving instance of the rewriting class,
mounted by a language model that sees the genuine choice, and we run
it against every scheme in the comparison
(Section~\ref{sec:attacks}).
\item We evaluate on ToolBench and ALFWorld against red--green and
multi-bit baselines: \method{} matches the unwatermarked agent's
success rate while the biased red--green watermark pays up to $8.1$
points, attributes at $1\%$ FPR under either single-axis attack where
a single rewriting pass erases both baselines, and concedes
essentially only the combined-attack corner, where the reseller has
already destroyed the service it resells (Section~\ref{sec:exp}).
\end{itemize}
\section{Related Work}\label{sec:related}

\paragraph{Watermarking LLM text.}
Statistical watermarks bias or derandomize token sampling with a
pseudorandom key. \citet{kirchenbauer2023watermark} boost the logits
of a keyed green list, trading detection power against distortion,
and \citet{zhao2023provable} harden the partition with a fixed
unigram key. Distortion-free and unbiased schemes remove the quality
cost: \citet{aaronson2023watermarking} couples sampling to a keyed Gumbel trick,
\citet{kuditipudi2023robust} formalize distortion-free samplers with
edit-robust detection, \citet{hu2024unbiased} and
\citet{christ2024undetectable} construct unbiased and
cryptographically undetectable variants, and
\citet{dathathri2024scalable} deploy sampling-level watermarking at
production scale. Multi-bit schemes carry identity payloads
\citep{yoo2024advancing}; BiMark \citep{feng2025bimark} stacks several
unbiased reweightings on each token, strengthening one signal on one
carrier, whereas our two layers place two signals on disjoint
carriers against two different attacks. A parallel line binds the
signal to sentence-level semantics so that it survives paraphrase
\citep{hou2024semstamp}; reliability studies and benchmarks map the
quality versus robustness frontier
\citep{kirchenbauer2024reliability,piet2025markmywords,liu2024survey}. On the
attack side, paraphrase removes token-level signals
\citep{krishna2023paraphrasing}, watermark stealing reverse-engineers
the keyed rules from API access \citep{jovanovic2024watermark},
\citet{sadasivan2023can} evade detectors through recursive
paraphrasing, and \citet{zhang2024watermarks} prove that a quality oracle
and a perturbation oracle suffice to erase any strong watermark.
These results shape our design rather than threaten it: the tally
channel's rewrite invariance is structural, not a statistical claim a
stronger paraphraser could erode (Theorem~\ref{thm:rewrite}), and for
the selection channel we prove what an informed rewriter achieves
(Proposition~\ref{prop:loser}). The selection channel transplants the
exponential race of \citet{kuditipudi2023robust} from tokens to
behaviors, where the candidate set is the environment's admissible
action set; the conditional law and entropy bound we prove for its
score (Lemma~\ref{lem:cond}, Theorem~\ref{thm:entropy}) appear to be
new even in the token setting.

\paragraph{Watermarking other generative modalities.}
Image watermarking has walked an arc from content-independent to
content-bound signals, and then to the granularity of the binding.
Regeneration attacks provably strip post-hoc invisible watermarks
\citep{zhao2024invisible}, pushing the signal into the generation
process itself; the resulting initial-noise schemes, Tree-Ring
\citep{wen2023tree}, Gaussian Shading \citep{yang2024gaussian}, and
pseudorandom codes \citep{gunn2025undetectable}, are training-free but
keyed independently of content, which black-box forgery exploits to
transplant a watermark onto arbitrary images \citep{muller2025black}.
Responses anchor detection in the initial noise itself
\citep{arabi2025hidden}, root the
signal in model weights or training data
\citep{fernandez2023stable,yu2021artificial}, or bind verification to
image semantics \citep{arabi2025seal}. The binding itself then became
the target: LLM-guided semantic injection defeats a single global
binding with edits that are locally substantial yet globally coherent
\citep{gao2026breaking}, and SLICE answers by anchoring distinct semantic
factors to disjoint regions of the initial noise \citep{gao2026slice}.
Behavioral trajectories invite the same lesson in a different
geometry: the selection channel binds the signal to trajectory content and
the tally channel to trajectory structure, and the two bindings fail
under complementary attacks by design.

\paragraph{Watermarking agent behavior.}
Tool-using agents
\citep{yao2022react,schick2023toolformer,qin2024toolllm,patil2024gorilla,Park2023GenerativeAI,Shridhar2020ALFWorldAT}
and the benchmarks that evaluate them as decision makers
\citep{Liu2023AgentBenchEL} expose a structured decision and observation
loop that text watermarks do not exploit, and a young line of work
embeds the signal there. Agent Guide \citep{huang2025agent}
biases the behavior distribution toward a keyed subset, a red--green
rule lifted to the decision level. AgentMark \citep{huang2026agentmark}
removes the bias through distribution-preserving conditional sampling
and carries a multi-bit identifier; it serves as a baseline in our
experiments. AgentWM \citep{wang2026protecting} targets model imitation,
biasing selection among semantically equivalent tool paths so that
the signal survives training on stolen outputs, and ActHook
\citep{meng2026watermarking} watermarks trajectory datasets with keyed hook
actions that a model trained on the data reproduces. Across this line
the adversary perturbs, paraphrases, or distills, and the watermark
is one signal under one keying. \method{} differs on the axis these
works leave open: the adversary who holds the log itself. We
formalize a reseller with lawful write access and two laundering
moves, and answer with a zero-bit, two-layer scheme carrying exact
finite-sample null distributions, unconditional rewrite invariance,
and a joint-erasure lower bound. The closest operational practice,
log signing, proves integrity of a log one already trusts but cannot
attribute an unsigned, possibly edited trajectory: a reseller
relaying a rebranded log simply drops the provider's signature.
Behavioral watermarking
therefore sits inside a broader provenance architecture, alongside
trusted execution environments, attestation, and authenticated logs,
as one complementary signal, the one that still speaks when the party
holding those records is itself the adversary. Our threat model (Section~\ref{sec:threat}) is built
for exactly this case, letting the reseller edit everything except
the environment-supplied action space and the executed action stream.
\section{Problem Formulation and Threat Model}\label{sec:prelim}

This section fixes, in the order a security argument needs them, the
object being watermarked (Section~\ref{sec:traj}), the parties and
their capabilities (Section~\ref{sec:threat}), and the attack classes
that every guarantee in this paper is stated against
(Section~\ref{sec:attacks}).

\subsection{Agent Trajectories}\label{sec:traj}

An agent interacts with an environment through alternating decisions
and observations; the log is modeled as a tagged sequence.

\begin{definition}[Agent trajectory]\label{def:traj}
An \emph{agent trajectory} is a finite sequence
$\tau := (e_1, \dots, e_T)$ of records, each carrying a role tag
$\rho(e_t) \in \{\textsc{dec}, \textsc{obs}\}$ and a content string
$c(e_t)$: \textsc{dec} records are decisions emitted by the agent,
\textsc{obs} records are observations returned by the environment. The
tag sequence $s(\tau) := (\rho(e_1), \dots, \rho(e_T))$ is the
\emph{skeleton} of $\tau$.
\end{definition}

These records log
emitted behavior, not the reasoning that produced it. A deployment
may expose three artifacts: internal reasoning traces (increasingly
hidden or encrypted), user-facing reasoning summaries (which may
diverge from the underlying computation), and the agent trajectory
just defined; \method{} reads and marks only the third.

\begin{definition}[Decision-boundary grouping]\label{def:group}
The \emph{decision-boundary grouping} of $\tau$ is the unique partition
of $\tau$ into consecutive blocks $g_1, \dots, g_m$, each consisting of
exactly one \textsc{dec} record followed by all \textsc{obs} records
preceding the next \textsc{dec} record (or the end of $\tau$); we write
$k_i \ge 0$ for the number of \textsc{obs} records in $g_i$.
\end{definition}

The grouping reads no textual markers and no content heuristics, only
the positions of \textsc{dec} tags; this gives it the invariance the
design will lean on.

\begin{proposition}[Skeleton determines positions]\label{prop:skeleton}
If $s(\tau) = s(\tau')$, then $\tau$ and $\tau'$ have the same number of
groups, the same group boundaries, and the same counts
$(k_1, \dots, k_m)$.
\end{proposition}

\begin{definition}[Group attributes]\label{def:attrs}
For each group $g_i$:
(i) $B_i$ is the \emph{candidate behavior set} presented to the agent
at decision $i$ by the environment (the available tools of a ToolBench
task, the admissible commands of an ALFWorld step; terminal actions
excluded), with $N_i := |B_i|$;
(ii) $P_i : B_i \to [0,1]$ is the agent's normalized distribution over
$B_i$, elicited at decision time, with uniform fallback;
(iii) $b_i \in B_i$ is the behavior actually selected; and
(iv) $A_i$ is the sequence of action identities recorded in $g_i$.
We write $\ent(P_i)$ for the Shannon entropy of $P_i$ in nats. Group
$g_i$ is \emph{effective} if $N_i \ge 2$ and the selection is
non-terminal, both predicates evaluated against the environment-supplied
$B_i$ and the executed action stream rather than any reseller-editable
record content, so the set of effective groups is itself invariant under
rewriting; only effective groups are pooled.
\end{definition}

Both layers draw randomness from one primitive, a deterministic random
bit generator $\drbg(\texttt{key}, \texttt{nonce}) \in [0,1)$
instantiated with HMAC-SHA512 \citep{Bellare1996KeyingHF,Barker2012RecommendationFR}:
embedder and detector recompute identical values from identical inputs,
and no side information is ever transmitted. The analysis adopts the
standard pseudorandom-function idealization
\citep{Goldreich1986HowTC}; independence holds only across distinct
evaluation points, so pooled detection deduplicates groups whose
evaluation points coincide.

\begin{assumption}[Ideal pseudorandomness]\label{ass:prf}
To any party not holding the key, the values
$\{\drbg(\texttt{key}, \nu)\}_{\nu}$ across distinct nonces $\nu$ are
i.i.d.\ uniform on $[0,1)$, independent of all other randomness; calls
under independent keys are mutually independent.
\end{assumption}

\subsection{Threat Model}\label{sec:threat}

\paragraph{Provider (defender).}
Two parties interact through the resale of an agentic service
(Figure~\ref{fig:threat}): a
\emph{provider}, who develops and serves the agent, and a
\emph{reseller}, who licenses that agent and relays it to customers
under a misrepresentation, as the reseller's own in-house system or as
a provider other than the one running underneath. The provider's
\emph{goal} is a test that, given a trajectory, decides whether its
agent produced it, at a false positive rate fixed in advance that no
key-less party can inflate, and at no cost to the service itself,
since a watermark that degrades task success will not be deployed;
formally, the provider, or an auditor acting with its keys, tests
$\Hnull$ (the trajectory was produced without knowledge of the keys)
against $\Halt$ (it was produced by the watermarked agent). Its
\emph{capability} is control of the sampler and the keys: it holds the
secret $(\keyone, \keytwo)$ and embeds \method{} at decision time
(Section~\ref{sec:method}), so every trajectory its agent produces
carries the watermark before leaving its control.

\paragraph{Reseller (adversary).}
The reseller's \emph{goal} is to defeat attribution: a trajectory the
provider's agent in fact produced must not be attributable to the
provider. Its defining \emph{capability} is full access to the
evidence: because it meters and bills usage, it stores and processes
every trajectory and may read and edit the log at will before any copy
reaches a verifier. Three limits bind it, and the design turns on all
three. \emph{No keys:} it does not hold $(\keyone, \keytwo)$, so it
cannot forge the watermark, only attempt to remove it.
\emph{No re-execution:} it resells the provider's live service, so the
executed action stream is genuine and its only lever is to edit the
recorded log.
\emph{Utility and consistency:} the edited log must remain useful and
consistent, since a service that no longer completes its customers'
tasks is worthless to resell, and the same logs substantiate the
reseller's charges, so they can be audited against the provider's
upstream record of consumption.

\subsection{Attack Classes}\label{sec:attacks}

Within these limits, editing the log decomposes into two operations,
which we formalize as the attack classes the analysis must withstand;
their composition is the \emph{combined attack}.
\emph{Pruning} the log, to shrink it or to strip records that name the
provider, is deletion.

\begin{definition}[Deletion attack]\label{def:del}
A \emph{deletion attack} drops records from the log. Because the
reseller cannot re-execute (Section~\ref{sec:threat}), the executed
\textsc{dec} records are preserved; the attack removes \textsc{obs}
records, lowering the counts $k_i$ of the affected groups (a group whose
redundant record is dropped falls from $k_i = 2$ to $k_i = 1$) and, in
the limit, reducing an affected group to its single preserved
\textsc{dec} record ($k_i = 0$). The
i.i.d.\ instance we sweep is parameterized by a rate $r \in [0,1]$: each
\textsc{obs} record is deleted independently with probability $r$.
Targeted removal of provider-identifying or redundant records,
truncation, and log loss are further instances.
\end{definition}

\emph{Rebranding} the log, paraphrasing observation text and renaming
actions into the reseller's namespace, is rewriting.

\begin{definition}[Rewriting attack]\label{def:rewrite}
A \emph{rewriting attack} is any map $R$ on trajectories with
$s(R(\tau)) = s(\tau)$: record contents (observation text and action
identities alike) may be edited arbitrarily, but no record is inserted
or deleted and no tag is changed.
\end{definition}

\paragraph{The LLM rewriter: a new informed attack.}
Definition~\ref{def:rewrite} is a class; evaluating against it requires
an instance, and prior work on behavioral watermarks tests only generic
perturbations such as random substitution or paraphrase
\citep{huang2025agent,huang2026agentmark,wang2026protecting}. We
introduce the \emph{LLM rewriter}, an attack we define here, to our
knowledge for the first time, that realizes the strongest rewrite a
reseller can mount without the keys: an \emph{informed},
\emph{plausibility-preserving} edit executed by a language model
(Figure~\ref{fig:llmrewriter}). For
each attacked group, an LLM receives the observation, the admissible
action set, and the action the agent actually chose, and returns a
\emph{different} action that a reasonable agent could plausibly have
taken; the recorded identity is replaced with its answer (exact prompt
in Appendix~\ref{app:algs}). The substitution is informed, since the
model sees the genuine choice, which is exactly the regime
Proposition~\ref{prop:loser} characterizes for the selection channel; and it
is plausibility-preserving, since an implausible replacement would
betray the edit to any reader and violate the consistency constraint
above. Applied to a fraction $q$ of groups, the LLM rewriter realizes
Definition~\ref{def:rewrite} at strength $q$ and leaves the skeleton
untouched by construction. Section~\ref{sec:exp} sweeps $q$ from $0$ to
$1$ and runs this attack against every scheme under comparison, ours
and the baselines alike.

\begin{wrapfigure}{r}{0.58\textwidth}
\vspace{-\baselineskip}
\centering
\includegraphics[width=0.56\textwidth]{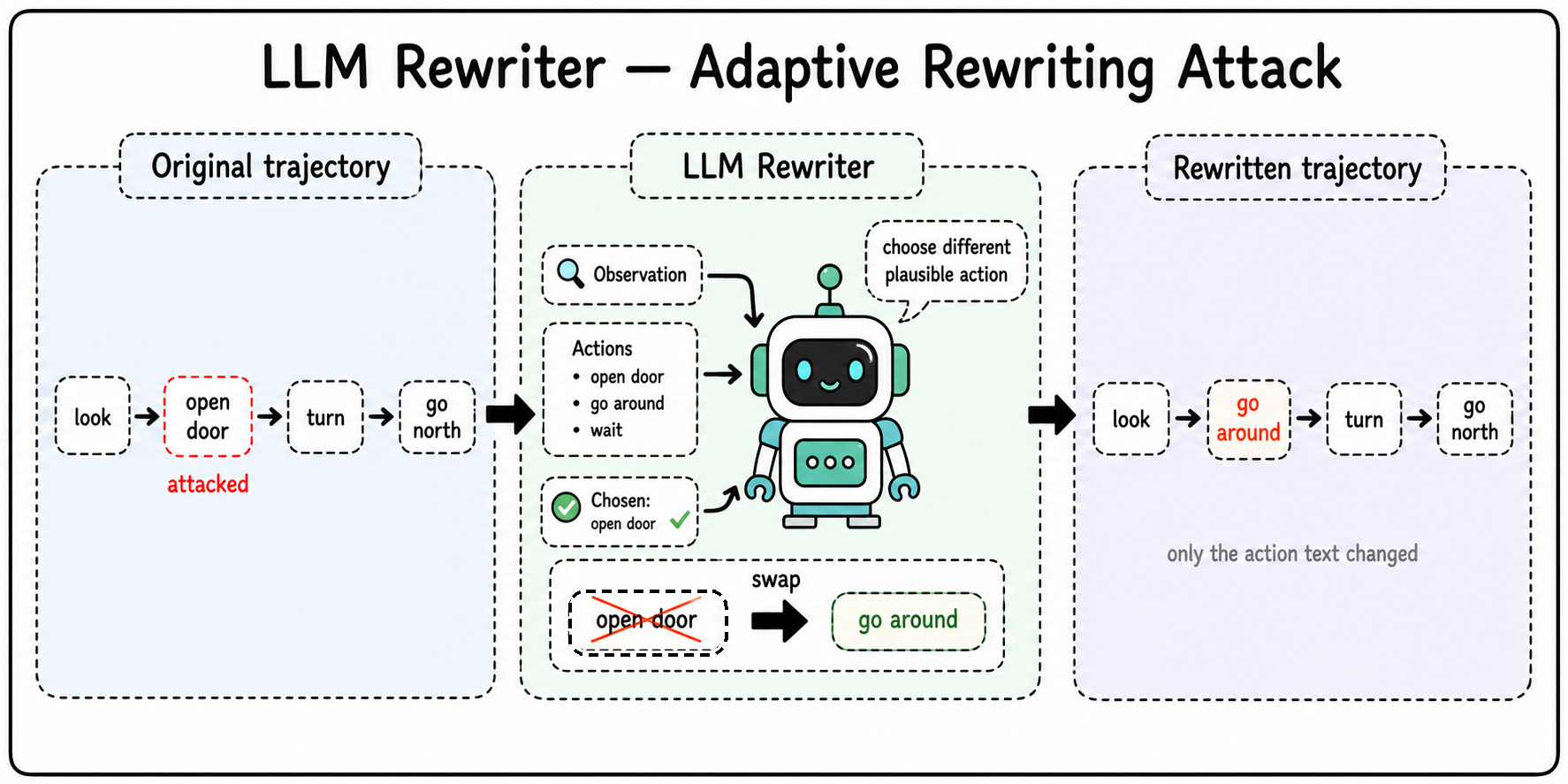}
\caption{The LLM rewriter swaps the recorded action for a different
plausible one, leaving the skeleton untouched.}
\label{fig:llmrewriter}
\vspace{-0.5\baselineskip}
\end{wrapfigure}

\paragraph{What the reseller cannot touch.}
Two things lie outside the rewriting class and beyond the reseller's
reach. The candidate behavior set $B_i$ is supplied by the environment,
not the log, so the verifier reads it at audit time from fields the
reseller cannot edit rather than
trusting it. And the executed action stream is what the provider
actually ran on the reseller's behalf: forging or deleting a
\textsc{dec} record breaks the correspondence between the log and that
stream, so attacks are required to act on the execution stream itself,
the data every detector under comparison reads. The detector accordingly
supports verification both from the reseller-released log and from the
grouping reconstructed from execution, the latter trusting no
reseller-editable field; large divergence between the two sources is
itself evidence of tampering, a \emph{log/execution consistency audit}
to which Section~\ref{sec:theory} returns.

\section{The \method{} Scheme}\label{sec:method}

\begin{figure}[!t]
\centering
\includegraphics[width=\textwidth]{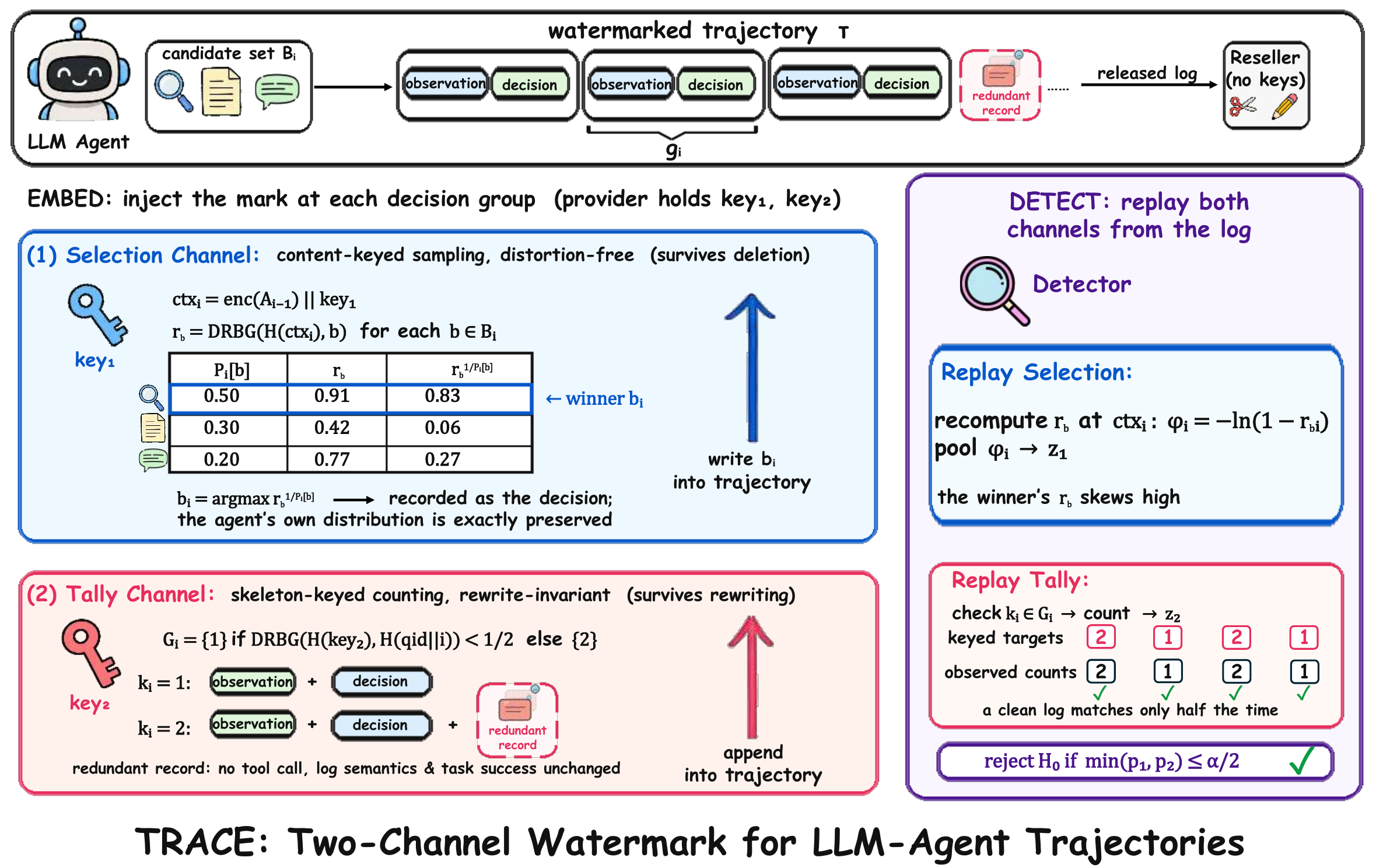}
\caption{Overview of \method{}: a content-keyed selection channel and a
skeleton-keyed tally channel embed two complementary marks at each
decision, read back by replaying both from the log.}
\label{fig:framework}
\end{figure}

Fix independent keys $\keyone, \keytwo$. \method{} consists of an
embedder, which replaces the agent's sampler at each decision and
appends keyed redundant records, and a detector, which maps an observed
trajectory to the statistics $(z_1, z_2)$ of \eqref{eq:z1} and
\eqref{eq:z2}; pseudocode for both is given as
Algorithms~\ref{alg:embed} and~\ref{alg:detect} in
Appendix~\ref{app:algs}. One principle governs the construction:
\emph{a statistic is invariant under an attack class as soon as its
carrier and its keying are functions of data preserved by every attack
in the class}. A rewriting attack preserves exactly the skeleton
(Proposition~\ref{prop:skeleton}); a deletion attack preserves the
contents of surviving groups and no positions; and no carrier--keying
pair is preserved by both classes, content failing under rewriting and
position under deletion. \method{} therefore runs one layer on each
invariant, both instances of the same template: keyed pseudorandomness
evaluated on invariant data, coupled to one carrier, detected by
replay (Figure~\ref{fig:framework}). The \emph{selection channel} modulates which action is
selected, keyed on content; the \emph{tally channel} modulates how many
records each decision group contains, keyed on position. Subscript $1$
refers throughout to the selection channel, $2$ to the tally channel;
Section~\ref{sec:theory} proves every property cited below.

\subsection{The Selection Channel: Content-Keyed Sampling}\label{sec:layer1}

The carrier is the selected behavior $b_i$. The watermark context of
group $i$ is
\begin{equation}\label{eq:ctx}
\ctx_i := \mathrm{enc}(A_{i-1}) \,\|\, \keyone,
\qquad
\ctx_1 := \textsc{bootstrap} \,\|\, \keyone,
\end{equation}
with $\mathrm{enc}$ an injective encoding of action-identity sequences,
and each candidate receives the value
\begin{equation}\label{eq:rb}
r_b := \drbg\bigl(\texttt{key} = H(\ctx_i),\ \texttt{nonce} = b\bigr)
\in (0,1), \qquad b \in B_i.
\end{equation}
Position is deliberately absent from \eqref{eq:ctx}: under deletion at
rate $r$, the context $\ctx_i$ is unchanged for every group whose
predecessor retains all its records, so only the immediate successors of
groups that lost a record evaluate the primitive at fresh points (blast
radius one) and the detector re-aligns at the next group, whereas a
single position in the keying path would desynchronize every group from
the first dropped record onward. The window has memory one because a longer
window, while equally admissible, widens the damage of each deletion
from one group to the window length (Remark~\ref{rem:window},
Appendix~\ref{app:practical}).

We realize the selection rule as a keyed exponential race over the
candidate set,
\begin{equation}\label{eq:exp}
b_i := \argmin_{b \in B_i} \frac{-\ln r_b}{P_i[b]}
     \;=\; \argmax_{b \in B_i}\ r_b^{\,1/P_i[b]},
\end{equation}
implemented in logarithmic form: every admissible behavior runs
against its own keyed clock, and the agent plays the winner. By
Theorem~\ref{thm:df} the race reproduces the agent's distribution
exactly, $\Pr[\,b_i = b\,] = P_i[b]$ for every $b \in B_i$, so utility
is settled by construction. The winner's value $r_{b_i}$ is stochastically
above uniform precisely because it won; this residue is the watermark.
The detector replays \eqref{eq:ctx} and \eqref{eq:rb} on the observed
trajectory, trusting no self-reported value, and scores each effective
group by
\begin{equation}\label{eq:phi}
\varphi_i := -\ln\bigl(1 - r_{b_i}\bigr),
\end{equation}
so that $\varphi_i \sim \mathrm{Exp}(1)$ exactly when $r_{b_i}$ is an
unwatermarked uniform. Pooled over the $n$ effective groups,
deduplicated so that no two share an evaluation point
(Assumption~\ref{ass:prf}),
\begin{equation}\label{eq:z1}
X_1 := \sum_{i=1}^{n} \varphi_i,
\qquad
z_1 := \frac{X_1 - n}{\sqrt{n}},
\end{equation}
with exact $p$-value $Q(n, X_1)$, the upper Gamma tail.
Section~\ref{sec:theory} and Appendix~\ref{app:proofs} supply the
quantitative content:
$X_1 \sim \mathrm{Gamma}(n, 1)$ under $\Hnull$ (Lemma~\ref{lem:null1});
$\E[\varphi_i] \ge 1 + \tfrac12 \ent(P_i)$ under $\Halt$
(Theorem~\ref{thm:entropy}); and under deletion at rate $r$ every score
generically stays at or above the null while at least a $(1-r)^2$ fraction of the entropy
signal survives, so the expected selection $z$ stays positive for every
$r < 1$ (Proposition~\ref{prop:blast}). Against
rewriting the layer claims nothing: the carrier itself is overwritten,
and replacing $b_i$ by a race loser never raises $\E[\varphi_i]$ above
$1$, driving it strictly below whenever the substitute carries positive
probability (Proposition~\ref{prop:loser}). Covering this gap is the
purpose of the tally channel.

\subsection{The Tally Channel: Skeleton-Keyed Counting}\label{sec:layer2}

Rewriting fixes the skeleton, hence the group count $m$ and the counts
$(k_1, \dots, k_m)$ (Proposition~\ref{prop:skeleton}); the tally channel
reads nothing else. The carrier is $k_i \in \{1, 2\}$, a zero-bit keyed
pattern, and the keying is a function of $(\keytwo, \texttt{qid}, i)$
alone:
\begin{equation}\label{eq:target}
\rho_i := \drbg\bigl(\texttt{key} = H(\keytwo),\
\texttt{nonce} = H(\texttt{qid} \,\|\, i)\bigr),
\qquad
G_i :=
\begin{cases}
\{1\}, & \rho_i < \tfrac12,\\
\{2\}, & \rho_i \ge \tfrac12,
\end{cases}
\end{equation}
with \texttt{qid} the task-instance identifier. Position, fatal in the
selection channel, is admissible here because rewriting cannot move it; an
earlier design keyed the target on $A_{i-1}$, and identity substitution
then desynchronized the recomputed targets from the embedded ones,
destroying the channel. The embedder realizes the target by
\[
k_i \;=\; 1 \;+\;
\mathbf{1}\bigl[\,G_i = \{2\}\,\bigr]\cdot
\mathbf{1}\bigl[\,\text{augmentation admissible for } g_i\,\bigr],
\]
appending, when both indicators equal $1$, one redundant record of the
following kind after the group's primary observation, a record that
disturbs neither the semantic content of the log nor the task's
execution; admissibility means that the environment's logging format
permits such a record.

\begin{definition}[Context-neutral redundant record]\label{def:redundant}
Let $\pi$ be the trajectory prefix up to and including the primary
observation of group $g_i$. A record $\tilde{e}$ appended to $g_i$ is
\emph{context-neutral} if
(i) $c(\tilde{e}) = f(\pi)$ for a fixed deterministic function $f$, so
that $\tilde{e}$ is informationally redundant with the prefix and the
agent's effective context is unchanged;
(ii) producing $\tilde{e}$ invokes no tool and incurs no environment
side effect; and
(iii) $\rho(\tilde{e}) = \textsc{obs}$, with no leading \textsc{dec}
record.
\end{definition}

Conditions (i) and (ii) are the utility guarantee: $c(\tilde{e})$ is a
deterministic function of a prefix the agent already possesses, so no
decision-relevant information changes downstream, and no tool is
re-executed (real agent tools are frequently side-effectful APIs). The
layer's resistance to forging or erasing the tally signal rests on
condition (iii) alone (Remark~\ref{rem:distinguish}), while its
invariance under rewriting is a separate guarantee, carried by the
skeleton keying (Theorem~\ref{thm:rewrite}).

\begin{remark}[Why no leading decision]\label{rem:distinguish}
A genuine extra tool call is necessarily headed by its own \textsc{dec}
record and therefore opens a \emph{new} group of count $1$ under
Definition~\ref{def:group}; the watermark's redundant record, having no
leading decision, folds into the \emph{current} group and raises its
count to $2$. The tally channel's signal is thus precisely ``a record
with no decision at its head.'' To forge or erase it the
reseller must insert or delete a \textsc{dec} record, an operation
outside the rewriting class of Definition~\ref{def:rewrite}, and one
that desynchronizes the log from the executed action stream.
\end{remark}

The detector replays \eqref{eq:target} on the observed skeleton,
trusting no self-reported field, and counts hits over the $n$ pooled
groups, whose index set is determined by the skeleton and the
environment-evaluated effectiveness of Definition~\ref{def:attrs} and
is therefore itself untouched by rewriting,
\begin{equation}\label{eq:z2}
X_2 := \sum_{i=1}^{n} \mathbf{1}\bigl[\,k_i \in G_i\,\bigr],
\qquad
z_2 := \frac{X_2 - n p_0}{\sqrt{n\, p_0 (1 - p_0)}},
\end{equation}
with exact Binomial $p$-value; here $p_0 = \tfrac12$ for every baseline
whose group counts satisfy $k_i \equiv 1$
(Appendix~\ref{app:practical} treats the general case). For a baseline
whose groups carry a single observation absent augmentation, the
watermarked agent with admissible augmentation attains $X_2 = n$, that
is, $z_2 = \sqrt{n}$, deterministically. The statistic $X_2$ is
invariant under every rewriting attack (Theorem~\ref{thm:rewrite}),
and $\Pr_{\Hnull}[X_2 = n] = 2^{-n}$ when every pooled group carries
one or two records, conservatively less otherwise
(Lemma~\ref{lem:null2}).

\subsection{The Composed Scheme}\label{sec:compose}

The coupling between the layers is one-way: when $k_i = 2$ the
augmented sequence $A_i$, the chosen action followed by the redundant
record's identity, enters the selection channel's window at group
$i + 1$, but
the tally channel reads no content, so the coupling graph is acyclic, a
single edge from the tally channel into the selection channel.
Consequently the exact $p$-values $p_1$ and $p_2$ are
independent under $\Hnull$ (Proposition~\ref{prop:compose}), and the
detector rejects when $\min(p_1, p_2) \le \alpha/2$, which bounds the
false positive rate by $\alpha$. Whether the selection channel's window
includes redundant records is a design knob that trades a sliver of
deletion exposure against decoupling; our experiments use the coupled
default (Remark~\ref{rem:window}, Appendix~\ref{app:practical}).

\section{Theoretical Analysis}\label{sec:theory}

This section states the four theorems that carry the paper's claims
and reads each one against the experiments of Section~\ref{sec:exp};
false-positive control, exact at every sample size, needs no theorem
of its own (Section~\ref{sec:theory-null}). The supporting lemmas and
propositions, together with all proofs, are deferred to
Appendix~\ref{app:proofs}, and every distributional claim was
additionally verified by Monte Carlo simulation
(Appendix~\ref{app:practical}). Appendix~\ref{app:notation} collects the
notation, and Appendix~\ref{app:resultsmap} maps the dependency
structure of every result below.

\subsection{Utility Preservation}\label{sec:theory-utility}

\begin{theorem}[Distortion-freeness]\label{thm:df}
Under Assumption~\ref{ass:prf}, the rule \eqref{eq:exp} satisfies
$\Pr[\,b_i = b\,] = P_i[b]$ for every $b \in B_i$.
\end{theorem}

\begin{takeaway}
The watermark is invisible in distribution at every decision: the keyed
race plays exactly the agent's own action distribution, not an
approximation of it, so the threat model's deployment constraint is met
by construction rather than by tuning.
\end{takeaway}

The proof, via the race lemma (Lemma~\ref{lem:race},
Appendix~\ref{app:race}), is the exponential-clocks form of
the Gumbel-max trick \citep{Maddison2014AS}. Table~\ref{tab:main} is
the theorem made visible: \method{} sits within seed noise of
\textsc{Base} on every benchmark, while the biased red--green
watermark pays $8.1$ points on ALFWorld ID, the price of moving
probability mass at low-entropy decisions that a distortion-free
sampler never moves.
The tally channel's utility guarantee is definitional rather than
distributional: by Definition~\ref{def:redundant}(i)--(ii) the appended
record is a deterministic function of context the agent already
possesses, invokes no tool, and incurs no side effect, so it changes no
decision-relevant information downstream.

\subsection{Exact False-Positive Control}\label{sec:theory-null}

Both detectors are exactly calibrated at every sample size, with no
asymptotics and no assumption beyond Assumption~\ref{ass:prf}. Under
$\Hnull$ the selection scores are i.i.d.\ $\mathrm{Exp}(1)$, so
$X_1 \sim \mathrm{Gamma}(n, 1)$ and its $p$-value is the upper Gamma
tail $Q(n, X_1)$ at every $n$ (Lemma~\ref{lem:null1}); the tally hits
are fair coins, so $X_2 \sim \mathrm{Bin}(n, \tfrac12)$, exactly so
when every pooled group carries one or two records and conservatively
otherwise, and a perfect hit count has probability exactly $2^{-n}$
(Lemma~\ref{lem:null2}).

The numbers compound fast: a watermarked agent forces every admissible
tally hit, so $30$ effective groups already certify provenance at a
false positive rate of $2^{-30} \approx 10^{-9}$. Forcing every hit also pins the alternative at $z_2 = \sqrt{n}$ by
\eqref{eq:z2}, the pooling rule visible in Table~\ref{tab:tb-detect},
where the $n = 50$ split's tally $z$ scales over the $n = 20$ splits
close to the predicted ratio ($1.45$ against
$\sqrt{2.5} \approx 1.58$), and \method{}'s wrong-key controls in
Table~\ref{tab:detect} are this null observed empirically: the tally
controls sit within noise of zero, and the selection controls sit
systematically below it, a direction that cannot inflate a one-sided
false positive rate. Rejecting when
$\min(p_1, p_2) \le \alpha/2$ bounds the combined false positive rate
by $\alpha$ via the union bound; the two channels' exact $p$-values
are moreover independent under $\Hnull$
(Proposition~\ref{prop:compose}), which licenses sharper combinations.

\subsection{The Entropy--Detectability Trade-off}\label{sec:theory-signal}

Under $\Halt$ the sampler \eqref{eq:exp} prefers candidates with large
$r_b$, so the replayed score sits stochastically above the null;
conditional on the winner $b$ it follows a generalized exponential law
with mean $\psi(1/p_b + 1) + \gamma$, where $\psi$ is the digamma
function and $\gamma$ the Euler--Mascheroni constant
(Lemma~\ref{lem:cond}; Corollaries~\ref{cor:uniform}
and~\ref{cor:pointmass} give its uniform and deterministic extremes).
Averaging that closed form yields the bound the channel turns on
(proof in Appendix~\ref{app:entropy}); to our knowledge it is new.

\begin{restatable}[Entropy lower bound on the signal]{theorem}{thmentropy}\label{thm:entropy}
For every distribution $P_i$,
\[
\E[\varphi_i] \;\ge\; 1 + \tfrac12\, \ent(P_i),
\]
with equality if and only if $P_i$ is a point mass.
\end{restatable}

\begin{takeaway}
Detection is paid for in entropy: every decision yields signal worth
at least half its entropy, and a deterministic decision yields none,
so a distortion-free watermark must pool short or low-entropy
trajectories rather than bias them.
\end{takeaway}

The
limit is broader than \method{}: a deterministic decision admits no
keyed variation under any distortion-free rule, so we read the entropy
price as intrinsic to distortion-free behavioral watermarking, and the
theorem as one sharp instance of it.
The trade-off is visible in the experiments. The detector pools groups
across trajectories at the explicit sample-complexity rate of
Corollary~\ref{cor:power}, and the two benchmarks separate just as the
bound's $\ent/2$ per-group rate predicts (Table~\ref{tab:detect}):
ToolBench's roughly $1.4$
effective groups per task yield selection $z \approx 4.5$ per split, while
ALFWorld's $23$ to $25$ decision steps accumulate the same guaranteed
per-group drift into $z$ between $94$ and $103$. The red--green watermark buys
its signal with probability mass instead and trails \method{} by $9.3$
success-rate points on ALFWorld ID (Table~\ref{tab:main}); \method{}
pays only in randomness the agent already spends.

\subsection{Robustness under the Threat Model}\label{sec:theory-robust}

Table~\ref{tab:ortho} is the design in one view: each layer's soft spot
is the other's strength.

\paragraph{Deletion: blast radius one.}
Deletion drops observation records, so it lowers the counts $k_i$ the
tally reads and corrupts that channel directly. The selection channel is
content-keyed: its carrier is the preserved decision $b_i$, so deletion
cannot substitute a race loser and every score stays at or above the
null, never below it save a rare context coincidence (the loser drift of
rewriting is impossible here).
Only the keying context can desynchronize, and the memory-one window
\eqref{eq:ctx} confines each dropped record's damage to a single
neighboring group, so at least a $(1-r)^2$ fraction of the entropy signal
survives at rate $r$ and the selection channel's expected score stays
positive for every $r < 1$ (Proposition~\ref{prop:blast}, Appendix~\ref{app:blast}).
This is the formal sense in which the selection channel is
self-synchronizing, and it is what Figure~\ref{fig:robust}(a) shows:
selection $z$ falls from $11.34$ to $4.06$ at $70\%$ deletion yet still
clears the threshold, while the count-keyed tally crosses it.

\paragraph{Rewriting: the selection channel inverts; the tally channel does not move.}
Against rewriting the selection channel claims nothing, and its failure is
sharper than erasure. Every identity the LLM rewriter substitutes is,
by construction, a candidate that lost the keyed race, so the replayed
score never rises above the null and falls strictly \emph{below} even
an unwatermarked log's whenever the substitute carried positive
probability under the agent's distribution
(Proposition~\ref{prop:loser}, Appendix~\ref{app:rewriteproofs}). That
is the below-null drift of Figure~\ref{fig:robust}(b), where selection $z$
crosses zero near $q \approx 0.5$ and ends at $-3.51$, and it is why no
amount of extra pooling rescues a content-keyed channel against a
rewriting reseller: pooling amplifies negative drift. A second,
rewrite-invariant channel is a necessity, not a preference, and its
guarantee is unconditional.

\begin{restatable}[Rewrite invariance of the tally]{theorem}{thmrewrite}\label{thm:rewrite}
For every rewriting attack $R$ (Definition~\ref{def:rewrite}) and every
trajectory $\tau$: $X_2(R(\tau)) = X_2(\tau)$, hence
$z_2(R(\tau)) = z_2(\tau)$.
\end{restatable}

\begin{takeaway}
Rewrite robustness here is an identity, not a bound: no rewriting
attack of any strength moves the tally statistic by a single bit,
because both its carrier and its key live in the skeleton that
rewriting, by definition, cannot touch.
\end{takeaway}

The invariance is
exactly as wide as Definition~\ref{def:rewrite}: semantic
rebranding (paraphrase and renaming in place) is covered, since it
preserves the tag sequence, whereas transformations that alter that
sequence, by inserting or deleting records, leave the class and are
met by Theorem~\ref{thm:erasure}(a) and the log/execution
consistency audit.
The experiments realize the identity to the digit: the tally $z$ stays
pinned at $14.34$ for every rewriting strength $q$ in
Figure~\ref{fig:robust}(b), with zero variation within each seed, while
both single-signal baselines collapse on the same axis. Against a
reseller who owns the log, an invariant the editor cannot move is the
one guarantee no rewriting effort can buy down. The headline claim is
then that erasing \emph{both} layers at once is qualitatively more
expensive than erasing either.

\begin{restatable}[Cost of joint erasure]{theorem}{thmerasure}\label{thm:erasure}
Let $\tau$ be watermarked with $m$ effective single primary-observation groups
($k_i = 1$ absent augmentation) and admissible augmentation throughout,
each group with $\ent(P_i) \ge h > 0$, and let
$A$ be any attack producing $\tau'$.
\begin{itemize}
\item[(a)] If $A$ is skeleton-preserving, then
$z_2(\tau') = z_2(\tau) = \sqrt{m}$: the tally channel is untouched. Hence any
attack with $\E[z_2(\tau')] < \sqrt{m}$ edits the skeleton, deleting
or inserting records, and is exposed to the log/execution consistency
audit.
\item[(b)] If $A$ corrupts the selection-channel evaluation point
(its context or its selected identity) of at most $a$ groups, whether by
deleting records that feed a group's context or by altering a selected
identity, and chooses the corrupted set \emph{obliviously}, that is,
independently of the realized selection-channel values (the
substitution \emph{within} a corrupted group may still be informed),
then, since the preserved decisions keep all $m$ groups,
\[
\E\bigl[z_1(\tau')\bigr]
\;\ge\;
\frac{(m - a)\bigl(1 + \tfrac{h}{2}\bigr) - m}{\sqrt{m}} .
\]
Consequently, for any threshold $\theta \ge 0$, $\E[z_1(\tau')] \le \theta$ forces
\[
a \;\ge\; \frac{h}{2 + h}\, m \;-\; \theta \sqrt{m}:
\]
the reseller must corrupt a constant fraction of the trajectory.
\end{itemize}
\end{restatable}

\begin{takeaway}
There is no cheap way to erase both channels: silencing the tally
channel means editing the skeleton and facing the consistency audit,
and silencing the selection channel means that any attack not
targeting groups by their realized
scores must corrupt a constant fraction of the very actions the resold
service depends on. Laundering the log means corrupting the product.
\end{takeaway}

\begin{table}[!t]
\centering
\caption{Orthogonal vulnerable surfaces: each channel is robust
precisely where the other breaks, with the cited results carrying the
guarantees.}
\label{tab:ortho}
\begin{tabular}{lll}
\toprule
Attack & Selection channel (carrier $b_i$) & Tally channel (carrier $k_i$) \\
\midrule
Deletion / dropped observations & robust: blast radius 1
(Prop.~\ref{prop:blast}) & broken (counts fall) \\
Rewriting / LLM rewriter & broken
(Prop.~\ref{prop:loser}) & invariant (Thm.~\ref{thm:rewrite}) \\
Combined (delete $+$ rewrite) & \multicolumn{2}{l}{cost bounded below
by Theorem~\ref{thm:erasure}} \\
\bottomrule
\end{tabular}
\end{table}

This is why only the joint high-deletion, high-substitution corner of
Figure~\ref{fig:heat} suppresses both detectors, and why neither cost
can be avoided by attacking only the released log, since detection
reads the grouping reconstructed from execution
(Section~\ref{sec:threat}). The oblivious-set hypothesis in (b) is
necessary rather than technical, since an attacker that targets groups
by their realized scores can do better, and the score-adaptive rate is
open (Remark~\ref{rem:adaptive}, Appendix~\ref{app:erasure}); the LLM
rewriter of Section~\ref{sec:exp} attacks a random fraction of groups
and so falls inside the covered regime.

\section{Experiments}\label{sec:exp}

We evaluate the three claims the theory makes: distortion-freeness
costs no utility (Section~\ref{sec:exp-utility}), detection is
calibrated and attributes the agent at deployment-relevant evidence
sizes (Section~\ref{sec:exp-detect}), and the two channels fail only
together, under exactly the attacks the threat model names
(Section~\ref{sec:exp-robust}).

\subsection{Experimental Setup}\label{sec:exp-setup}

\paragraph{Schemes.}
Four arms run on every benchmark. \textsc{Base} is the unwatermarked
agent. \textsc{AM-F} is AgentMark \citep{huang2026agentmark}, the
closest existing agent watermark: distribution-preserving multi-bit
embedding whose payload is recovered through random linear network
coding (RLNC). \textsc{RG} is the red--green watermark of
\citet{kirchenbauer2023watermark} lifted from tokens to behaviors,
biasing selection toward a keyed green subset of $B_i$
($\gamma = 0.5$, $\delta = 2.0$) under the same content window as
\method{}'s selection channel; it represents the canonical biased,
single-signal design point. \method{} is the scheme of
Section~\ref{sec:method}, its keys expanded through SHA-256 into the
material driving the HMAC-SHA512 DRBG.

\paragraph{Data.}
ToolBench \citep{qin2024toolllm} evaluates tool-use decision making
over a large corpus of real-world APIs. We use its six test splits,
T1 to T6, spanning the single-tool, intra-category multi-tool, and
intra-collection multi-tool regimes under held-out instructions,
tools, and categories ($n = 20$ per seed and split, except T5 with
$n = 50$); its trajectories are short, roughly $1.4$ effective
decision groups per task. ALFWorld \citep{Shridhar2020ALFWorldAT}
evaluates embodied household planning in interactive text
environments. Its six task families, A1 to A6, ask the agent to find,
process, and place objects: simple pick-and-place, cleaning, heating,
or cooling an item before putting it away, examining an object under a
lamp, and placing two objects; trajectories are long, $23$ to $25$
decision steps. We use ALFWorld's in-distribution validation split (ID, $140$
tasks) and its out-of-distribution split of unseen environments (OOD,
$134$ tasks). Task subsets are fixed and shared across arms and seeds,
all results are mean $\pm$ sample standard deviation over three seeds,
and the label-to-split mapping is tabulated in
Table~\ref{tab:splits} (Appendix~\ref{app:exp}). The backbone is GPT-5.4-mini, served through
an API relay; an ALFWorld ablation on a second, locally deployed
Qwen backbone is in Appendix~\ref{app:qwen}
(Tables~\ref{tab:qwen-util} and~\ref{tab:qwen-detect}).

\paragraph{Metrics.}
\emph{1) Utility:} success rate (SR), scored by one LLM judge
(GPT-5.4-mini, identical across arms; solved $1$, unsure $0.5$,
unsolved $0$), and steps per task; \method{} step counts include the
tally channel's redundant records, so they measure the full logged
overhead. \emph{2) Detection:} the pooled per-channel $z$ of
Section~\ref{sec:method}, thresholded at $\theta = 2$ and combined
across channels by $\max(z_1, z_2)$, the normal-approximation
counterpart of Algorithm~\ref{alg:detect}'s exact rule; wrong-key
controls (wk) rerun each detector under an independent key never used
at embedding, estimating the empirical null. Because the threat model
fixes the verifier's false positive rate in advance, we also report
\emph{TPR at calibrated FPR}: positives are watermarked trajectories
scored under the true key, negatives the same trajectories under the
wrong key; bundles of $B$ trajectories are pooled into one statistic,
and TPR@$x\%$FPR is the fraction of positive bundles above the
$(1{-}x\%)$ quantile of the negative bundles ($5000$ bootstrap
resamples). Under attack, each attack cell already pools $118$ to
$127$ trajectories into one $z$, so a cell is one fixed-size bundle
and
TPR@$1\%$FPR is the fraction of cells above the one-sided Gaussian
$1\%$ threshold $2.326$, calibrated by the wrong-key per-cell null.
Together these measure what a verifier holding the evidence can
attribute at its chosen FPR (protocol details in
Appendix~\ref{app:tpr}). AM-F is additionally scored on its native
metric, verified channel bits per task.

\paragraph{Attacks.}
The two classes of Section~\ref{sec:attacks}, applied to every arm:
random observation deletion at rate $r$ (five trials per seed), our LLM
rewriter at strength $q$ (three trials per seed), and their
composition, rewriting at strength $q$ followed by deletion at rate
$r$, over a $3 \times 6$ grid of $(r, q)$ with nine runs per cell; the
detection threshold is $\theta = 2$ throughout.

\begin{table}[!t]
\centering
\caption{Main results on the principal backbone: per-task success rate
and steps, mean $\pm$ std over three seeds (Avg.\ rows $n$-weighted
with deltas vs.\ \textsc{Base}; green: no degradation beyond seed
noise, red: clear drop; $\dagger$: no pooled std; \method{} steps
include the tally channel's redundant records).}
\label{tab:main}
\vspace{2mm}
\resizebox{\textwidth}{!}{%
\begin{tabular}{l|c|cccc|cccc}
\toprule
\multirow{2}{*}{\textbf{Setting}} & \multirow{2}{*}{\textbf{Task}}
& \multicolumn{4}{c|}{\textbf{SR (\%)} $\uparrow$}
& \multicolumn{4}{c}{\textbf{Steps / task}} \\
\cmidrule(lr){3-6} \cmidrule(lr){7-10}
& & Base & AM-F & RG & \method{}
& Base & AM-F & RG & \method{} \\
\midrule\midrule
\cellcolor{red!5} & A1 & $91.4{\pm}2.9$ & $93.3{\pm}4.4$ & $89.5{\pm}1.6$ & $91.4{\pm}2.9$
 & $18.1{\pm}2.3$ & $17.9{\pm}1.6$ & $17.8{\pm}0.9$ & $24.7{\pm}2.7$ \\
\cellcolor{red!5} & A2 & $85.2{\pm}3.7$ & $84.0{\pm}4.3$ & $81.5{\pm}7.4$ & $88.9{\pm}7.4$
 & $24.3{\pm}3.2$ & $23.5{\pm}1.6$ & $26.8{\pm}3.2$ & $30.1{\pm}5.4$ \\
\cellcolor{red!5} & A3 & $81.2{\pm}12.5$ & $89.6{\pm}3.6$ & $75.0{\pm}6.2$ & $93.8{\pm}6.2$
 & $25.1{\pm}3.1$ & $23.9{\pm}3.7$ & $30.6{\pm}0.9$ & $37.8{\pm}2.7$ \\
\cellcolor{red!5} & A4 & $76.0{\pm}4.0$ & $80.0{\pm}4.0$ & $57.3{\pm}8.3$ & $77.3{\pm}4.6$
 & $27.3{\pm}3.3$ & $24.6{\pm}1.7$ & $33.6{\pm}2.7$ & $40.7{\pm}3.3$ \\
\cellcolor{red!5} & A5 & $87.2{\pm}4.4$ & $92.3{\pm}7.7$ & $82.1{\pm}4.4$ & $94.9{\pm}8.9$
 & $19.6{\pm}0.8$ & $17.7{\pm}5.1$ & $21.6{\pm}2.8$ & $18.9{\pm}4.2$ \\
\cellcolor{red!5} & A6 & $70.8{\pm}0.0$ & $66.7{\pm}7.2$ & $56.9{\pm}9.6$ & $59.7{\pm}8.7$
 & $35.9{\pm}1.0$ & $32.7{\pm}1.8$ & $37.6{\pm}0.9$ & $51.7{\pm}2.2$ \\
\cellcolor{red!5}\multirow{-7}{*}{\textbf{\shortstack[l]{ALFWorld\\ID}}} & \textbf{Avg.}
 & $\mathbf{82.4{\pm}2.1}$
 & \cellcolor{green!66!black!20}$\mathbf{84.0{\pm}3.0}\,(\uparrow 1.6)$
 & \cellcolor{lightred}$\mathbf{74.3{\pm}2.9}\,(\downarrow 8.1)$
 & \cellcolor{green!66!black!20}$\mathbf{83.6{\pm}2.6}\,(\uparrow 1.2)$
 & $\mathbf{24.9}^{\dagger}$
 & $\mathbf{23.4}^{\dagger}\,(\downarrow 1.5)$
 & $\mathbf{27.6}^{\dagger}\,(\uparrow 2.7)$
 & $\mathbf{34.2{\pm}2.6}\,(\uparrow 9.3)$ \\
\midrule
\cellcolor{orange!7} & A1 & $90.3{\pm}4.8$ & $98.6{\pm}2.4$ & $80.6{\pm}2.4$ & $90.3{\pm}6.4$
 & $21.6{\pm}3.2$ & $15.2{\pm}0.4$ & $23.9{\pm}2.6$ & $29.6{\pm}2.4$ \\
\cellcolor{orange!7} & A2 & $87.1{\pm}5.6$ & $84.9{\pm}4.9$ & $82.8{\pm}4.9$ & $78.5{\pm}8.1$
 & $20.7{\pm}2.4$ & $23.1{\pm}1.0$ & $22.9{\pm}3.1$ & $38.9{\pm}4.3$ \\
\cellcolor{orange!7} & A3 & $85.5{\pm}2.5$ & $91.3{\pm}4.3$ & $76.8{\pm}6.6$ & $76.8{\pm}5.0$
 & $24.1{\pm}4.1$ & $21.9{\pm}2.9$ & $30.7{\pm}0.3$ & $41.9{\pm}3.2$ \\
\cellcolor{orange!7} & A4 & $90.5{\pm}0.0$ & $88.9{\pm}2.7$ & $87.3{\pm}5.5$ & $90.5{\pm}8.2$
 & $21.7{\pm}0.9$ & $19.8{\pm}1.5$ & $20.3{\pm}4.3$ & $27.9{\pm}3.5$ \\
\cellcolor{orange!7} & A5 & $88.9{\pm}5.6$ & $83.3{\pm}5.6$ & $83.3{\pm}0.0$ & $92.6{\pm}3.2$
 & $23.0{\pm}1.2$ & $19.0{\pm}1.5$ & $23.9{\pm}1.6$ & $27.9{\pm}4.9$ \\
\cellcolor{orange!7} & A6 & $33.3{\pm}9.0$ & $45.1{\pm}12.2$ & $51.0{\pm}6.8$ & $54.9{\pm}14.8$
 & $43.1{\pm}0.8$ & $42.4{\pm}0.9$ & $40.0{\pm}0.6$ & $55.6{\pm}6.6$ \\
\cellcolor{orange!7}\multirow{-7}{*}{\textbf{\shortstack[l]{ALFWorld\\OOD}}} & \textbf{Avg.}
 & $\mathbf{81.3{\pm}2.2}$
 & \cellcolor{green!66!black!20}$\mathbf{83.8{\pm}2.2}\,(\uparrow 2.5)$
 & \cellcolor{lightred}$\mathbf{78.1{\pm}4.1}\,(\downarrow 3.2)$
 & \cellcolor{green!66!black!20}$\mathbf{81.1{\pm}4.2}\,(\downarrow 0.2)$
 & $\mathbf{24.8}^{\dagger}$
 & $\mathbf{22.9}^{\dagger}\,(\downarrow 1.9)$
 & $\mathbf{26.3}^{\dagger}\,(\uparrow 1.5)$
 & $\mathbf{36.7{\pm}2.1}\,(\uparrow 11.9)$ \\
\midrule
\cellcolor{LavenderLight!20} & T1 & $72.5{\pm}10.9$ & $83.3{\pm}2.9$ & $78.3{\pm}2.9$ & $81.7{\pm}2.9$
 & $2.18{\pm}0.33$ & $1.97{\pm}0.13$ & $1.98{\pm}0.08$ & $1.80{\pm}0.22$ \\
\cellcolor{LavenderLight!20} & T2 & $85.0{\pm}5.0$ & $86.7{\pm}5.8$ & $78.3{\pm}12.6$ & $86.7{\pm}7.6$
 & $2.17{\pm}0.13$ & $2.00{\pm}0.22$ & $2.10{\pm}0.13$ & $1.93{\pm}0.13$ \\
\cellcolor{LavenderLight!20} & T3 & $78.3{\pm}7.6$ & $75.0{\pm}0.0$ & $76.7{\pm}10.4$ & $83.3{\pm}5.8$
 & $2.17{\pm}0.13$ & $2.07{\pm}0.10$ & $2.33{\pm}0.24$ & $2.37{\pm}0.34$ \\
\cellcolor{LavenderLight!20} & T4 & $91.7{\pm}7.6$ & $90.0{\pm}0.0$ & $88.3{\pm}2.9$ & $90.0{\pm}10.0$
 & $2.38{\pm}0.25$ & $2.28{\pm}0.19$ & $2.17{\pm}0.03$ & $2.40{\pm}0.66$ \\
\cellcolor{LavenderLight!20} & T5 & $69.3{\pm}5.0$ & $73.3{\pm}1.2$ & $75.3{\pm}3.1$ & $65.7{\pm}5.9$
 & $2.12{\pm}0.22$ & $2.01{\pm}0.16$ & $2.07{\pm}0.20$ & $1.81{\pm}0.05$ \\
\cellcolor{LavenderLight!20} & T6 & $78.3{\pm}2.9$ & $73.3{\pm}2.9$ & $71.7{\pm}2.9$ & $68.3{\pm}11.5$
 & $2.27{\pm}0.12$ & $1.97{\pm}0.08$ & $2.02{\pm}0.16$ & $2.35{\pm}0.90$ \\
\cellcolor{LavenderLight!20}\multirow{-7}{*}{\textbf{ToolBench}} & \textbf{Avg.}
 & $\mathbf{77.2{\pm}1.8}$
 & \cellcolor{green!66!black!20}$\mathbf{78.9{\pm}0.4}\,(\uparrow 1.7)$
 & \cellcolor{green!66!black!20}$\mathbf{77.6{\pm}3.4}\,(\uparrow 0.4)$
 & \cellcolor{green!66!black!20}$\mathbf{76.6{\pm}2.2}\,(\downarrow 0.6)$
 & $\mathbf{2.20{\pm}0.10}$
 & $\mathbf{2.04{\pm}0.05}\,(\downarrow 0.16)$
 & $\mathbf{2.10{\pm}0.08}\,(\downarrow 0.10)$
 & $\mathbf{2.05{\pm}0.16}\,(\downarrow 0.15)$ \\
\bottomrule
\end{tabular}}
\end{table}

\subsection{Task Utility under Watermarking}\label{sec:exp-utility}

Table~\ref{tab:main} is the empirical face of Theorem~\ref{thm:df}. On
ToolBench every arm sits within about two points of \textsc{Base} in
aggregate, where \method{}'s weighted success rate is indistinguishable
from \textsc{Base} at seed noise; the wider per-split swings (most
visibly T5 and T6) sit within the larger seed variance on those splits. On
ALFWorld \method{} matches \textsc{Base} in aggregate on both splits
($+1.2$\,pp ID, $-0.2$\,pp OOD, within seed noise), while the biased
\textsc{RG} pays $8.1$ points ID and $3.2$ OOD; the per-type rows
locate the damage where low-entropy decisions concentrate, most
visibly A4 ID. \textsc{AM-F}, also distribution-preserving, stays at
\textsc{Base} level, confirming that what RG pays for is the bias
itself. \method{}'s step counts include the tally channel's redundant
records, reported on the same all-episode denominator as every other
arm, and decomposing them shows where the gap to \textsc{Base} lives:
the redundant records contribute $11.0$ (ID) and $11.9$ (OOD) entries
per task, while the decision path itself runs $23.2$ and $24.8$ steps
against \textsc{Base}'s $24.9$ and $24.8$, at parity. The entire
overhead is watermark records that by Definition~\ref{def:redundant}
invoke no tool and incur no environment side effect, log lines rather
than agent work; on ToolBench the same accounting adds about $0.7$
redundant records on top of $1.4$ decision groups per task. The full
accounting, per split and per backbone, is tabulated in
Table~\ref{tab:redundant} (Appendix~\ref{app:redundant}). \textbf{Distortion-freeness is free in
practice: \method{} tracks the unwatermarked agent within seed noise
on every benchmark, while the biased watermark pays its utility tax
exactly where decisions are low-entropy.}

\begin{table}[!t]
\centering
\caption{Detection on the principal backbone: pooled $z$ per channel
with wrong-key (wk) controls and AM-F capacity in verified channel
bits ($\ddagger$:
mean over the six per-split statistics of Table~\ref{tab:tb-detect}).}
\label{tab:detect}
\footnotesize
\vspace{2mm}
\setlength{\tabcolsep}{3.6pt}
\begin{tabular}{l|cc|cccc|c}
\toprule
\multirow{2}{*}{\textbf{Setting}} & \multicolumn{2}{c|}{\textbf{RG}} & \multicolumn{4}{c|}{\textbf{\method{}}}
& \multicolumn{1}{c}{\textbf{AM-F}} \\
\cmidrule(lr){2-3} \cmidrule(lr){4-7} \cmidrule(lr){8-8}
& $z$ $\uparrow$ & wk & Sel.\ $z$ $\uparrow$ & Tally $z$ $\uparrow$
& wk sel. & wk tally & bits/task $\uparrow$ \\
\midrule\midrule
\cellcolor{red!5}ToolBench & $2.72^{\ddagger}$ & $-0.88^{\ddagger}$ & \sbest{$4.51^{\ddagger}$}
 & \best{$5.77^{\ddagger}$} & $-0.16^{\ddagger}$ & $0.43^{\ddagger}$
 & $1.68{\pm}0.06$ \\
\cellcolor{orange!7}ALFWorld ID & $37.37{\pm}0.58$ & $-0.04{\pm}1.29$ & \best{$94.15{\pm}6.16$}
 & \sbest{$54.38{\pm}2.22$} & $-4.53{\pm}1.62$ & $1.06{\pm}0.48$
 & $53.85{\pm}1.30$ \\
\cellcolor{orange!7}ALFWorld OOD & $34.48{\pm}1.31$ & $1.99{\pm}0.55$ & \best{$102.53{\pm}6.16$}
 & \sbest{$55.32{\pm}1.74$} & $-2.43{\pm}0.67$ & $0.73{\pm}0.29$
 & $49.34{\pm}2.21$ \\
\bottomrule
\end{tabular}
\end{table}

\begin{figure}[!t]
\centering
\includegraphics[width=0.58\textwidth]{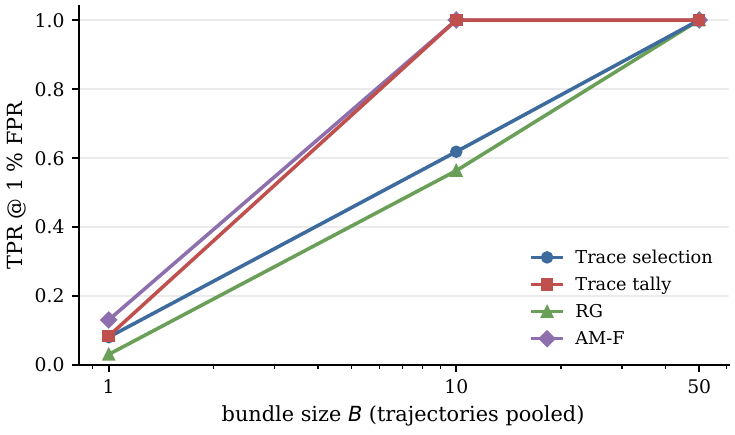}
\caption{TPR at $1\%$ FPR on ToolBench against the wrong-key null, as
a function of the number of pooled trajectories $B$ (ALFWorld and
per-FPR breakdowns in Appendix~\ref{app:tpr}).}
\label{fig:tpr}
\end{figure}

\subsection{Detection Power and Evidence Requirements}\label{sec:exp-detect}

Table~\ref{tab:detect} reports pooled detection, and \method{}'s
wrong-key controls behave as Section~\ref{sec:theory-null} demands:
the tally controls sit within noise of zero, and the selection
controls sit systematically \emph{below} it, the harmless side of a
one-sided test. The two benchmarks separate exactly along
Theorem~\ref{thm:entropy}: ToolBench's $1.4$ effective groups per task
yield single-digit per-split $z$, while ALFWorld's long horizons push
the selection channel to $z$ between $94$ and $103$ and the tally into
the fifties
(Table~\ref{tab:detect}; per-seed values in Table~\ref{tab:alf-seed}),
well above RG. AM-F's verified channel bits
follow the same horizons, $1.68$ per ToolBench task
against $49$ to $54$ on ALFWorld: capacity, like our signal, is bought
with decision entropy.

Figure~\ref{fig:tpr} asks the question a verifier actually faces: how
many trajectories buy attribution at a fixed false positive rate? On
ToolBench a single trajectory is not enough for \emph{any} behavioral
watermark: its trajectories are simply too short, and $1.4$ effective
decisions carry too little entropy, the price
Theorem~\ref{thm:entropy} fixes for every scheme that leaves the
agent's distribution intact. \method{} converts evidence into
attribution as fast as anything we tested. On the long-horizon
ALFWorld a \emph{single} trajectory already attributes most of the
time, TPR@$1\%$FPR reaching $0.87$ to $0.94$ for the selection channel
and $0.85$ to $0.86$ for the tally (full tables in
Appendix~\ref{app:tpr}, Tables~\ref{tab:tpr-clean}
and~\ref{tab:tpr-fpr}); on ToolBench, ten pooled trajectories take
the tally channel to TPR $1.000$, volume the resale setting supplies
by definition. AM-F's clean detection keeps pace at $B = 10$, but it
buys that power with a payload design that
Section~\ref{sec:exp-robust} shows collapsing under a single rewriting
pass, while the tally channel does not move. \textbf{A single
long-horizon trajectory, or ten short ones, attributes \method{}'s
agent at $1\%$ FPR, and \method{} alone keeps this power under every
single-axis attack that follows.}

\begin{figure}[!t]
\centering
\includegraphics[width=\textwidth]{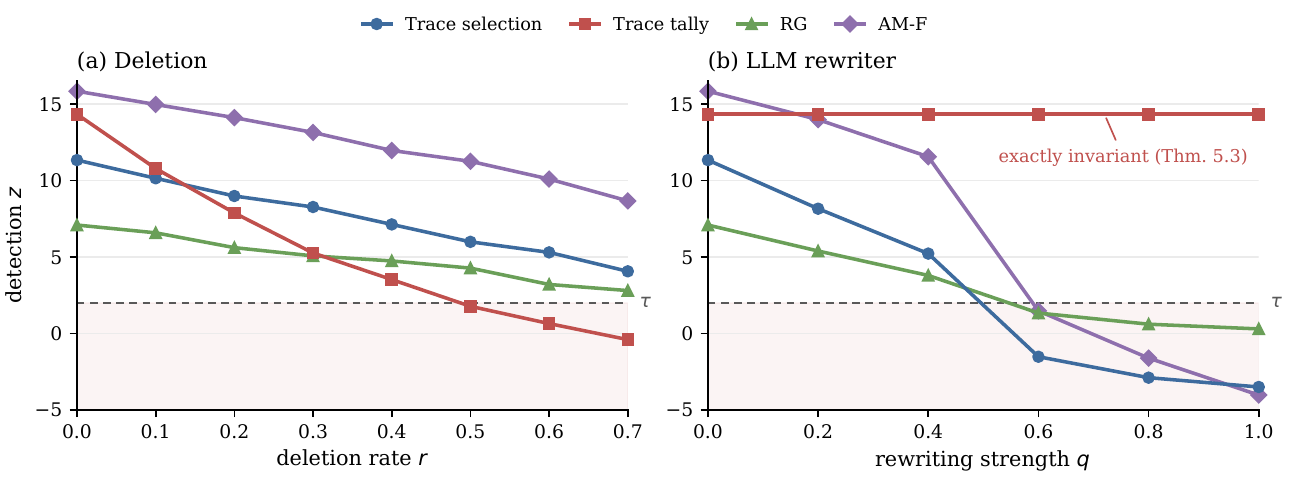}
\caption{Detection $z$ on ToolBench under (a) observation deletion at rate
$r$ and (b) our LLM rewriter at strength $q$ (shaded: below the
threshold $\theta = 2$).}
\label{fig:robust}
\end{figure}

\begin{figure}[!t]
\centering
\includegraphics[width=\textwidth]{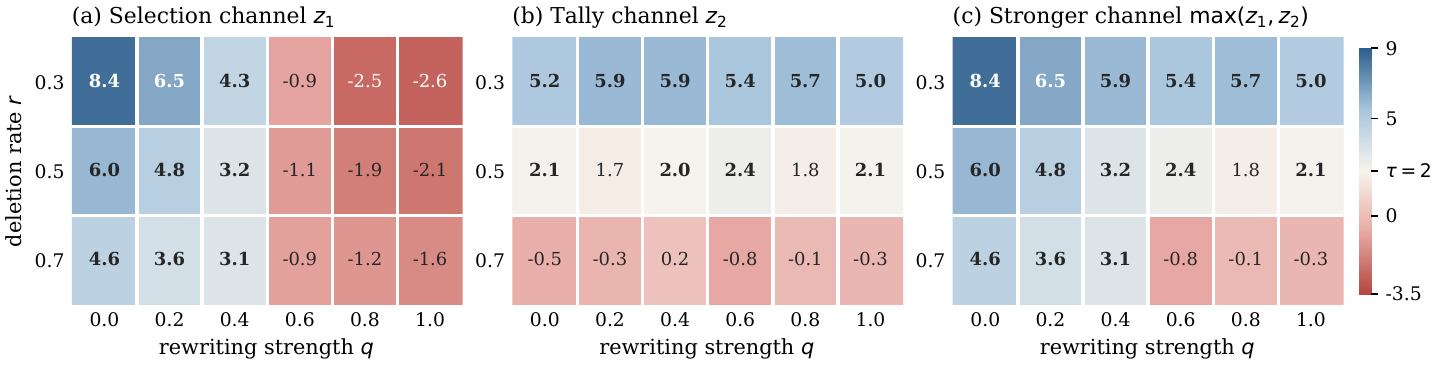}
\caption{Combined attack on \method{} (ToolBench): per-channel
detection $z$ and, in panel (c), the stronger channel, over deletion
rate $r$ and rewriting strength $q$, with colors diverging at
$\theta = 2$.}
\label{fig:heat}
\end{figure}

\subsection{Robustness under Deletion, Rewriting, and Combined Attacks}\label{sec:exp-robust}

The reseller of Section~\ref{sec:threat} launders the log with two
moves, pruning and rebranding, and its goal is to push detection below
$\theta$ while keeping the log sellable. Table~\ref{tab:ortho} predicts
exactly how this plays out: each attack class breaks one of \method{}'s
channels by design and cannot touch the other, whereas every baseline
carries a single signal and therefore owns a single fatal axis. This
subsection tests that prediction move by move.

\paragraph{Deletion (Figure~\ref{fig:robust}a).}
Dropping observations lowers the counts $k_i$ the tally reads, so the
count-keyed tally channel must fail, and it does, crossing $\theta$ near
$r \approx 0.5$.
The question deletion actually poses is whether anything survives it,
and the selection channel does: its carrier is the preserved decision,
so deletion can only desynchronize the keying context, and the
memory-one window confines each dropped record's damage to a single
neighboring group (Proposition~\ref{prop:blast}),
so the detector resynchronizes after every cut, and $z$ decays
gracefully from $11.34$ to $4.06$ at $70\%$ deletion, still above
threshold, with TPR@$1\%$FPR never below $0.93$ at any tested rate
(Table~\ref{tab:tpr-robust}). To
even reach that point, the reseller has already discarded seventy
percent of the service it bills for. The baselines also weather this
axis: RG, content-keyed like our selection channel but with a weaker
per-group signal, decays from $7.09$ to $2.80$; AM-F rides its RLNC
erasure coding to $z = 8.67$ at bit accuracy $1.00$ (full sweep in
Table~\ref{tab:base-del}). Deletion alone launders nothing.

\paragraph{Rewriting (Figure~\ref{fig:robust}b).}
Rebranding, renaming actions into the reseller's own namespace, is the
move a reseller mounts as a matter of course, and our LLM rewriter is
its strongest informed instance. It erases every content-keyed signal
in the comparison: RG falls from $7.09$ to $0.29$, AM-F from $15.85$
to $-4.04$ with bit accuracy $0.34$, and both baselines' TPR collapses
from $1.000$ to $0.000$ by $q = 0.8$ (full sweep in
Table~\ref{tab:base-sub}). The schemes that kept pace with
\method{} on clean logs are erased by a single rewriting pass.
\method{}'s own selection channel fails on this axis too, by design,
and even \emph{inverts}: every substituted identity is a race loser,
so the replayed score sinks below the null, crossing zero near
$q \approx 0.5$ and ending at $-3.51$, the drift
Proposition~\ref{prop:loser} predicts. The tally channel, whose
carrier and key live in the skeleton the rewriter cannot touch by
definition, does not move at all: $z = 14.34$ at every strength, with
zero variation across strengths within every seed, TPR $1.000$
throughout,
Theorem~\ref{thm:rewrite} realized to the digit. This is the gap
between empirical robustness and structural invariance: no rewriting
budget, however large, buys the reseller anything against the tally.

\paragraph{Combined attack (Figure~\ref{fig:heat}).}
The only move left is to delete and rewrite at once. Panels (a) and
(b) confirm that the two channels fail orthogonally, $q$ draining the
selection channel and $r$ the tally, nearly independently of one
another. Panel (c) shows what the combined detector sees, the stronger
of the two channels, and it stays above $\theta$ everywhere except the
corner $r = 0.7$ with $q \ge 0.6$ (plus the marginal cell
$(0.5, 0.8)$ at $z = 1.81$), exactly the regime
Theorem~\ref{thm:erasure} prices. Reaching that corner means rewriting
most of the steps \emph{and} deleting $70\%$ of the result: a log
that no longer resembles the service the reseller bills for, carrying
the skeleton edits the log/execution consistency audit flags. The RG
and AM-F grids (Table~\ref{tab:base-comb}; TPR versions in
Table~\ref{tab:tpr-comb}) need no corner: each is
already erased by one axis alone. \textbf{Every single-signal scheme
has a fatal axis; \method{} has none: under either single-axis attack
one channel keeps attributing at $1\%$ FPR, and silencing both costs
the reseller the very product it resells.}

\section{Conclusion}\label{sec:conclusion}

We introduced \method, a two-layer behavioral watermark for LLM agents
in which a distortion-free, content-keyed selection channel (robust to
deletion, blast radius one) and a position-keyed tally channel
(unconditionally invariant to rewriting) are superposed on one
trajectory with independent keys and one-way coupling.
Against an adversary holding the evidence,
attribution should decompose across complementary invariants, each
keyed to data one attack class must preserve, a principle we
conjecture extends beyond watermarking. The analysis
gives exact null distributions for both detectors, a closed-form
conditional law for the selection score, an entropy lower bound that
makes the trade-off between utility and detectability of
distortion-free agent watermarking precise, and a joint-erasure theorem
showing that suppressing both layers forces skeleton edits and, for
obliviously targeted attacks, constant-fraction deletion or alteration
simultaneously. Experiments on
ToolBench and ALFWorld bear the predictions out, including the exact
invariance of the tally channel under substitution and the below-null drift of
the selection channel under informed rewriting. Limitations and
directions: detection requires the verifier to access the candidate
sets, which assumes the environment's action space is available at
audit time; the two channels buy their robustness at complementary
prices, the selection channel paying in entropy, its power degrading
on near-deterministic agents exactly as Theorem~\ref{thm:entropy}
predicts with pooling requirements growing accordingly, and the tally
channel paying in log volume, its redundant records adding eleven to
twelve entries per ALFWorld task, an overhead in storage and audit
length rather than in tool calls or task success, since the decision
path stays at parity with the unwatermarked agent
(Section~\ref{sec:exp-utility}); and our guarantees concern
\emph{removal}, so a
reseller that abandons the provider's agent entirely and re-runs the
task on a different model produces a genuinely unwatermarked trajectory,
whose missing mark flags the misrepresentation but does not by itself
name the substitute. Extending the tally channel beyond $\{1,2\}$,
watermarking multi-agent interaction patterns,
and treating watermark removal as the attacker's
joint optimization of attribution score, task utility, and
stealth are natural next steps.

\paragraph{Impact statement.}
This work targets accountability infrastructure for autonomous agents,
specifically the resale setting in which a middleman controls the logs.
Watermarks of this kind can deter the misattribution of agent behavior
and support audit trails; conversely, any provenance tool can in
principle be used to track benign users of an agent system. Our scheme
watermarks the provider's own trajectories with the provider's keys and
reveals nothing about third parties; we believe the accountability
benefits outweigh the risks.

\bibliographystyle{plainnat}
\bibliography{references}

\appendix

\addtocontents{toc}{\protect\setcounter{tocdepth}{2}}
\clearpage
\begingroup
\hypersetup{linkcolor=black}
\renewcommand{\contentsname}{Table of Contents}
\tableofcontents
\endgroup
\clearpage

\section*{Appendix}

\section{Notation}\label{app:notation}

We collect the symbols used throughout. The conventions are those of
Sections~\ref{sec:prelim} and~\ref{sec:method}; this table is a
reference, not a redefinition, and every entry points to the place
where the object is introduced.

\begin{center}
\small
\renewcommand{\arraystretch}{1.2}
\begin{longtable}{@{}p{0.30\linewidth}p{0.64\linewidth}@{}}
\caption{Symbols used throughout \method{}, each with a pointer to where
it is introduced.}\label{tab:notation}\\
\toprule
\textbf{Symbol} & \textbf{Meaning} \\
\midrule
\endfirsthead
\toprule
\textbf{Symbol} & \textbf{Meaning} \\
\midrule
\endhead
\midrule
\multicolumn{2}{r}{\footnotesize\itshape continued on next page} \\
\endfoot
\bottomrule
\endlastfoot
\multicolumn{2}{@{}l}{\textit{Trajectories and grouping} (Section~\ref{sec:traj})} \\[2pt]
$\tau$ & agent trajectory: a finite sequence of tagged records (Def.~\ref{def:traj}) \\
$e_t,\ \rho(e_t),\ c(e_t)$ & record $t$, its role tag $\rho\in\{\textsc{dec},\textsc{obs}\}$, its content string (Def.~\ref{def:traj}) \\
$s(\tau)$ & skeleton: the tag sequence $(\rho(e_1),\dots,\rho(e_T))$ (Def.~\ref{def:traj}) \\
$g_1,\dots,g_m$ & decision-boundary grouping into $m$ groups (Def.~\ref{def:group}) \\
$k_i$ & number of \textsc{obs} records in $g_i$; the tally channel's carrier (Def.~\ref{def:group}) \\
$B_i,\ N_i=|B_i|$ & candidate behavior set at decision $i$ and its size (Def.~\ref{def:attrs}) \\
$P_i,\ b_i,\ A_i$ & agent distribution over $B_i$; the selected behavior; the recorded action-identity sequence (Def.~\ref{def:attrs}) \\
$\ent(P_i)$ & Shannon entropy of $P_i$ in nats (Def.~\ref{def:attrs}) \\
effective group & a group with $N_i\ge 2$ and non-terminal selection; only these are pooled (Def.~\ref{def:attrs}) \\
$n$ & number of effective, deduplicated pooled groups \\[4pt]
\multicolumn{2}{@{}l}{\textit{Keyed pseudorandomness} (Section~\ref{sec:traj})} \\[2pt]
$\drbg(\texttt{key},\texttt{nonce})$ & deterministic random bit generator in $[0,1)$, instantiated with HMAC-SHA512 \\
$H(\cdot)$ & hash deriving \textsc{drbg} keys and nonces from structured inputs \\
$\keyone,\ \keytwo$ & the provider's two independent secret keys \\
Assumption~\ref{ass:prf} & ideal pseudorandomness: i.i.d.\ uniform across distinct nonces, independent across keys \\[4pt]
\multicolumn{2}{@{}l}{\textit{Selection channel} (Section~\ref{sec:layer1})} \\[2pt]
$\ctx_i=\mathrm{enc}(A_{i-1})\,\|\,\keyone$ & content-keyed watermark context of group $i$ \eqref{eq:ctx} \\
$\mathrm{enc}$ & injective encoding of action-identity sequences \\
$r_b=\drbg(H(\ctx_i),b)$ & keyed value in $(0,1)$ for candidate $b$ \eqref{eq:rb} \\
$b_i=\argmax_b r_b^{1/P_i[b]}$ & distortion-free keyed exponential race \eqref{eq:exp} \\
$\varphi_i=-\ln(1-r_{b_i})$ & selection score of group $i$ \eqref{eq:phi} \\
$X_1,\ z_1$ & pooled selection sum $\sum_i\varphi_i$ and its $z$-statistic \eqref{eq:z1} \\
$Q(n,X_1)$ & upper regularized incomplete Gamma tail: the selection $p$-value \\
$p_1$ & exact selection-channel $p$-value \\[4pt]
\multicolumn{2}{@{}l}{\textit{Tally channel} (Section~\ref{sec:layer2})} \\[2pt]
$\rho_i=\drbg(H(\keytwo),H(\texttt{qid}\,\|\,i))$ & skeleton-keyed target draw for group $i$ \eqref{eq:target} \\
$G_i\in\{\{1\},\{2\}\}$ & keyed target count set \eqref{eq:target} \\
$\texttt{qid}$ & task-instance identifier \\
$\tilde e$ & context-neutral redundant record appended when $G_i=\{2\}$ (Def.~\ref{def:redundant}) \\
$p_0=\tfrac12$ & tally null hit rate \\
$X_2,\ z_2$ & tally hit count $\sum_i \mathbf{1}[k_i\in G_i]$ and its $z$-statistic \eqref{eq:z2} \\
$p_2$ & exact tally-channel $p$-value \\
$\theta$ & detection threshold on the pooled $z$-statistic ($\theta=2$ in the experiments) \\[4pt]
\multicolumn{2}{@{}l}{\textit{Threat model and attacks} (Sections~\ref{sec:threat}--\ref{sec:attacks})} \\[2pt]
provider / reseller / verifier & defender; log-holding adversary; key-holding auditor (Section~\ref{sec:threat}) \\
$r$ & per-observation deletion rate: each \textsc{obs} record dropped independently w.p.\ $r$ (Def.~\ref{def:del}) \\
$R,\ q$ & rewriting attack with $s(R(\tau))=s(\tau)$, at strength $q$ (Def.~\ref{def:rewrite}) \\
LLM rewriter & informed, plausibility-preserving substitution instance (Section~\ref{sec:attacks}) \\[4pt]
\multicolumn{2}{@{}l}{\textit{Detection and constants} (Section~\ref{sec:theory})} \\[2pt]
$\Hnull,\ \Halt$ & null (key-less) and alternative (watermarked) hypotheses (Section~\ref{sec:threat}) \\
$\alpha$ & target false positive rate; the detector rejects when $\min(p_1,p_2)\le\alpha/2$ \\
$\psi,\ \psi',\ \gamma$ & digamma, trigamma, and the Euler--Mascheroni constant (Lem.~\ref{lem:cond}) \\
generalized exponential law & conditional law of $\varphi_i$ given the winner, shape $1/p_b$ (Lem.~\ref{lem:cond}) \\
\end{longtable}
\end{center}

\section{Map of the Results}\label{app:resultsmap}

Figure~\ref{fig:resultsmap} traces the dependency structure of the
theory in four layers. Two restated facts supply the inputs the
construction stands on: the ideal-pseudorandomness assumption
(Assumption~\ref{ass:prf}) and the skeleton invariance of the grouping
(Proposition~\ref{prop:skeleton}, immediate from
Definition~\ref{def:group}). The lemmas are the load-bearing
mechanisms: the race lemma (Lemma~\ref{lem:race}) for distortion-freeness,
the conditional law of the score (Lemma~\ref{lem:cond}) and the trigamma
estimate (Lemma~\ref{lem:trigamma}) for the entropy bound, and the two
null laws (Lemmas~\ref{lem:null1} and~\ref{lem:null2}) for exact
false-positive control. The four boxed theorems assemble these:
distortion-freeness (Theorem~\ref{thm:df}) from the race lemma, the
entropy--detectability bound (Theorem~\ref{thm:entropy}) from the
conditional law and the trigamma estimate, rewrite invariance
(Theorem~\ref{thm:rewrite}) from skeleton invariance, and the
robustness propositions, blast radius one
(Proposition~\ref{prop:blast}) and below-null drift under informed
substitution (Proposition~\ref{prop:loser}), from the entropy bound and
the conditional law respectively. Everything converges on the capstone:
the joint-erasure theorem (Theorem~\ref{thm:erasure}) reads off rewrite
invariance, the tally null, the entropy bound, and the two robustness
propositions, pricing the cost of silencing both channels at once.
Off to the side, one-way coupling (Proposition~\ref{prop:compose})
combines the two null laws under the cross-key independence of
Assumption~\ref{ass:prf} to license the joint test, and the
sample-complexity corollary (Corollary~\ref{cor:power}) turns the
entropy bound, together with the conditional-law variance, into the
pooling rate.

\begin{figure}[!t]
\centering
\begin{tikzpicture}[
  >={Stealth[length=2mm]},
  every node/.style={font=\footnotesize},
  res/.style={draw, rounded corners=2pt, align=center, inner sep=3pt,
              fill=LightBlue!18, minimum height=8mm, text width=24mm},
  capstone/.style={draw, rounded corners=2pt, align=center, inner sep=3pt,
              fill=bestbg, minimum height=8mm, text width=24mm},
  ext/.style={draw, dashed, rounded corners=2pt, align=center, inner sep=3pt,
              fill=gray!8, minimum height=8mm, text width=24mm},
  dep/.style={->, gray!70, shorten >=1pt, shorten <=1pt},
]
\node[ext] (ap) at (0,1.5)   {Assumption~\ref{ass:prf}\\\scriptsize ideal PRF};
\node[ext] (sk) at (0,-3.4)  {Prop.~\ref{prop:skeleton}\\\scriptsize skeleton fixes positions};
\node[res] (race) at (3.8,3.4)  {Lemma~\ref{lem:race}\\\scriptsize race property};
\node[res] (cond) at (3.8,2.0)  {Lemma~\ref{lem:cond}\\\scriptsize conditional law};
\node[res] (trig) at (3.8,0.7)  {Lemma~\ref{lem:trigamma}\\\scriptsize trigamma bound};
\node[res] (n1)   at (3.8,-0.7) {Lemma~\ref{lem:null1}\\\scriptsize selection null};
\node[res] (n2)   at (3.8,-2.2) {Lemma~\ref{lem:null2}\\\scriptsize tally null};
\node[res] (df)   at (7.6,3.4)  {Thm~\ref{thm:df}\\\scriptsize distortion-free};
\node[res] (ent)  at (7.6,2.0)  {Thm~\ref{thm:entropy}\\\scriptsize entropy bound};
\node[res] (blast)at (7.6,0.7)  {Prop.~\ref{prop:blast}\\\scriptsize blast radius one};
\node[res] (loser)at (7.6,-0.7) {Prop.~\ref{prop:loser}\\\scriptsize below-null drift};
\node[res] (comp) at (7.6,-2.0) {Prop.~\ref{prop:compose}\\\scriptsize joint FPR};
\node[res] (rw)   at (7.6,-3.4) {Thm~\ref{thm:rewrite}\\\scriptsize rewrite invariance};
\node[res]      (pow)  at (11.4,2.0)  {Cor.~\ref{cor:power}\\\scriptsize detection power};
\node[capstone] (eras) at (11.4,-1.4) {Thm~\ref{thm:erasure}\\\scriptsize joint erasure};
\draw[dep] (ap) -- (race);
\draw[dep] (ap) -- (n1);
\draw[dep] (ap) -- (n2);
\draw[dep] (ap) -- (comp);
\draw[dep] (sk) -- (rw);
\draw[dep] (race) -- (df);
\draw[dep] (race) -- (cond);
\draw[dep] (cond) -- (ent);
\draw[dep] (cond) -- (loser);
\draw[dep] (trig) -- (ent);
\draw[dep] (ent) -- (pow);
\draw[dep] (cond) to[out=22,in=158] (pow);
\draw[dep] (ent) -- (blast);
\draw[dep] (ent) -- (eras);
\draw[dep] (n1) -- (comp);
\draw[dep] (n2) -- (comp);
\draw[dep] (blast) -- (eras);
\draw[dep] (loser) -- (eras);
\draw[dep] (rw) -- (eras);
\draw[dep] (n2) to[out=-25,in=200] (eras);
\end{tikzpicture}
\caption{Dependency map of the theory. Dashed boxes are restated
foundational facts; solid boxes are results of this paper; the yellow
box is the capstone. Arrows trace each result's principal proof inputs;
secondary and purely transitive dependencies are omitted for clarity.
Utility and detection (top) and robustness (bottom) meet at the
joint-erasure theorem.}
\label{fig:resultsmap}
\end{figure}
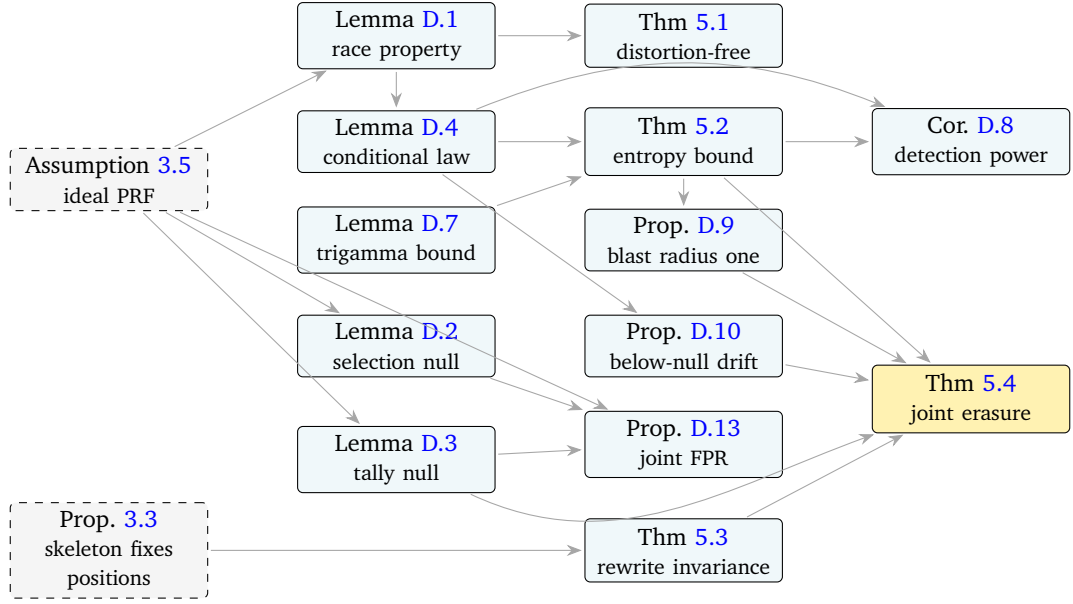

\section{Algorithms and the LLM Rewriter Prompt}\label{app:algs}

Algorithms~\ref{alg:embed} and~\ref{alg:detect} give the embedder and
detector of Section~\ref{sec:method} in full. The prompt below
instantiates the LLM rewriter, the informed rewriting attack we
introduce in Section~\ref{sec:attacks}; one call is issued per attacked
group.

\begin{algorithm}[!ht]
\caption{\method{} embedding (one task instance)}
\label{alg:embed}
\begin{algorithmic}[1]
\State $A_0 \gets \textsc{bootstrap}$
\For{each decision $i = 1, 2, \dots$}
  \State elicit $B_i$, $P_i$ from agent and environment
  \State $\ctx_i \gets \mathrm{enc}(A_{i-1}) \,\|\, \keyone$;\quad
         $r_b \gets \drbg(H(\ctx_i), b)$ for all $b \in B_i$
  \State $b_i \gets \argmin_b\, (-\ln r_b)/P_i[b]$
         \Comment{selection channel: distortion-free selection}
  \State execute $b_i$; record decision and primary observation
  \State $\rho_i \gets \drbg(H(\keytwo), H(\texttt{qid} \| i))$
  \If{$\rho_i \ge \tfrac12$ \textbf{and} augmentation admissible}
    \State append context-neutral redundant record $\tilde e$ with
           $c(\tilde e) = f(\pi)$ \Comment{tally channel: $k_i = 2$}
  \EndIf
  \State $A_i \gets$ action identities recorded in $g_i$
\EndFor
\end{algorithmic}
\end{algorithm}

\begin{algorithm}[!ht]
\caption{\method{} detection (pooled over trajectories)}
\label{alg:detect}
\begin{algorithmic}[1]
\State regroup each trajectory by its skeleton (Definition~\ref{def:group});
       never trust self-reported indices or values
\State deduplicate: among groups sharing an evaluation point, keep one
       representative (Assumption~\ref{ass:prf})
\For{each effective group $i$}
  \State recompute $\ctx_i$, $(r_b)_{b \in B_i}$, and
         $\varphi_i \gets -\ln(1 - r_{b_i})$ \Comment{selection channel}
  \State recompute $G_i$ from \eqref{eq:target}; record hit
         $\mathbf{1}[k_i \in G_i]$ \Comment{tally channel}
\EndFor
\State $z_1 \gets (X_1 - n)/\sqrt{n}$ with exact $p$-value $Q(n, X_1)$
\State $z_2 \gets (X_2 - n p_0)/\sqrt{n p_0 (1-p_0)}$ with exact
       Binomial $p$-value
\State \Return reject $\Hnull$ if $\min(p_1, p_2) \le \alpha/2$
\end{algorithmic}
\end{algorithm}

\begin{promptbox}[frametitle={LLM rewriter prompt}]
\small\ttfamily\raggedright
{\normalfont\bfseries System.}\enspace You audit AI-agent
trajectories. Given one decision step, you pick a different but still
plausible action the agent could have taken.

\medskip
{\normalfont\bfseries User (template, one call per attacked
group).}\enspace At this step the agent observed:\\
\{obs\}\\
Available actions: \{actions\}\\
The agent chose: \{chosen\}\\
Pick ONE different action from the available list that a reasonable
agent could plausibly have taken instead. Reply with only the exact
action name, nothing else.
\end{promptbox}

\FloatBarrier
\section{Supporting Theory and Missing Proofs}\label{app:proofs}

This appendix states the supporting results deferred from
Section~\ref{sec:theory} and proves every claim of the paper.
Throughout, $\psi$ is the digamma function, $\psi' = \psi^{(1)}$ the
trigamma function, and $\gamma = -\psi(1)$ the Euler--Mascheroni
constant. Proposition~\ref{prop:skeleton} is immediate from
Definition~\ref{def:group}: the partition depends only on the positions
of \textsc{dec} tags. The remaining results follow in order of
appearance.

\subsection{The Race Lemma and Proof of Theorem~\ref{thm:df}}\label{app:race}

\begin{lemma}[Race property]\label{lem:race}
Let $E_b := -\ln r_b$, $T_b := E_b / P_i[b]$, and
$T := \min_b T_b$. Then $b_i$ and $T$ are independent,
$\Pr[\,b_i = b\,] = P_i[b]$, and $T \sim \mathrm{Exp}(1)$.
\end{lemma}

\begin{proof}
The $E_b$ are i.i.d.\ $\mathrm{Exp}(1)$ by Assumption~\ref{ass:prf}, so
$T_b \sim \mathrm{Exp}(P_i[b])$ independently (we restrict to the support
$\{b : P_i[b] > 0\}$; a candidate with $P_i[b] = 0$ sets $T_b = +\infty$
and is selected with probability $0$). For $t \ge 0$,
\[
\Pr[\,b_i = b,\ T > t\,]
= \int_t^\infty P_i[b]\, e^{-P_i[b] s}
  \prod_{b' \ne b} e^{-P_i[b'] s}\, ds
= P_i[b] \int_t^\infty e^{-s}\, ds
= P_i[b]\, e^{-t},
\]
using $\sum_{b'} P_i[b'] = 1$. The product form gives independence and
both marginals at once. Theorem~\ref{thm:df} is immediate.
\end{proof}

\subsection{Exact Null Distributions}\label{app:null}

\begin{lemma}[Selection-channel null distribution]\label{lem:null1}
Under $\Hnull$ and Assumption~\ref{ass:prf}, the scores
$\varphi_1, \dots, \varphi_n$ of the $n$ pooled (effective,
deduplicated) groups are i.i.d.\ $\mathrm{Exp}(1)$; hence
$X_1 \sim \mathrm{Gamma}(n, 1)$ exactly, the one-sided $p$-value is the
regularized upper incomplete Gamma function $Q(n, X_1)$, and
$z_1 \Rightarrow \mathcal{N}(0,1)$ as $n \to \infty$.
\end{lemma}

\begin{proof}
Under $\Hnull$ the producer holds no information about $\keyone$, so by
Assumption~\ref{ass:prf} the value $r_{b_i}$ at the realized
(deduplicated) evaluation point $(\ctx_i, b_i)$ is uniform on $(0,1)$,
independently across groups. Then $\varphi_i \sim \mathrm{Exp}(1)$, and
the remaining claims are standard properties of sums of i.i.d.\
exponentials.
\end{proof}

\begin{lemma}[Tally null distribution and exact tail]\label{lem:null2}
Under $\Hnull$ with $p_0 = \tfrac12$, if every pooled group has
$k_i \in \{1,2\}$ then $X_2 \sim \mathrm{Bin}(n, \tfrac12)$; in
particular $\Pr[X_2 = n] = 2^{-n}$. A group with $k_i \notin \{1,2\}$
is a forced miss, so for a general baseline $X_2$ is stochastically
dominated by $\mathrm{Bin}(n, \tfrac12)$ and the upper tail remains a
conservative $p$-value.
\end{lemma}

\begin{proof}
The pooled groups' tally evaluation points $H(\texttt{qid} \,\|\, i)$
are pairwise distinct, distinct task instances carrying distinct
\texttt{qid} and within-trajectory indices being distinct, with any
residual collision removed by the deduplication of
Assumption~\ref{ass:prf}. Under $\Hnull$ the targets $G_i$ are
therefore, by Assumption~\ref{ass:prf}, i.i.d.\ uniform over
$\{\{1\},\{2\}\}$ and independent of the trajectory. When $k_i \in \{1,2\}$, exactly one of $\{1\},\{2\}$
contains $k_i$, so each indicator in \eqref{eq:z2} is a fair coin; a
group with $k_i \notin \{1,2\}$ lies in neither target and is a
deterministic miss, which only lowers $X_2$
(Appendix~\ref{app:practical}).
\end{proof}

\subsection{The Conditional Law of the Score}\label{app:cond}

\begin{lemma}[Conditional law of the score]\label{lem:cond}
Let group $i$ be watermarked via \eqref{eq:exp} and write
$p_b := P_i[b]$. Conditional on $b_i = b$, the score $\varphi_i$ has
distribution function
\[
F_b(t) = \bigl(1 - e^{-t}\bigr)^{1/p_b}, \qquad t \ge 0,
\]
the generalized exponential law with shape $1/p_b$
\citep{Gupta1999TheoryM}. Consequently
\[
\E[\varphi_i \mid b_i = b] = \psi\!\Bigl(\tfrac{1}{p_b} + 1\Bigr) + \gamma,
\qquad
\Var[\varphi_i \mid b_i = b]
= \psi'(1) - \psi'\!\Bigl(\tfrac{1}{p_b} + 1\Bigr) < \frac{\pi^2}{6},
\]
and, unconditionally,
$\E[\varphi_i] = \sum_{b} p_b \bigl(\psi(1/p_b + 1) + \gamma\bigr)$.
\end{lemma}

\begin{proof}
Write $p := p_b$ and condition throughout on $b_i = b$. By
Lemma~\ref{lem:race}, the winning exponential satisfies
$E_{b_i} = p\, T$ where $T := \min_{b'} T_{b'} \sim \mathrm{Exp}(1)$ is
independent of the identity of the winner. Hence
$r_{b_i} = e^{-E_{b_i}} = e^{-pT}$ and
\[
\varphi_i = -\ln\bigl(1 - e^{-pT}\bigr), \qquad T \sim \mathrm{Exp}(1).
\]
For $t \ge 0$,
\[
\Pr[\varphi_i \le t]
= \Pr\bigl[e^{-pT} \le 1 - e^{-t}\bigr]
= \Pr\Bigl[T \ge -\tfrac{1}{p}\ln(1 - e^{-t})\Bigr]
= \bigl(1 - e^{-t}\bigr)^{1/p},
\]
which is the claimed distribution function with shape
$\alpha := 1/p \ge 1$.

For the moments we compute the moment generating function. The density
is $f(t) = \alpha (1 - e^{-t})^{\alpha - 1} e^{-t}$, and for $s < 1$ the
substitution $v = e^{-t}$ (so $e^{st} = v^{-s}$, $dt = -dv/v$) gives
\[
\E\bigl[e^{s \varphi_i}\bigr]
= \alpha \int_0^\infty e^{st} (1 - e^{-t})^{\alpha-1} e^{-t}\, dt
= \alpha \int_0^1 v^{-s} (1 - v)^{\alpha - 1}\, dv
= \alpha\, B(1 - s, \alpha)
= \frac{\Gamma(1-s)\, \Gamma(\alpha + 1)}{\Gamma(\alpha + 1 - s)} .
\]
The cumulant generating function is therefore
$K(s) = \ln \Gamma(1-s) + \ln \Gamma(\alpha+1) - \ln \Gamma(\alpha+1-s)$,
with
\[
K'(s) = -\psi(1-s) + \psi(\alpha + 1 - s),
\qquad
K''(s) = \psi'(1-s) - \psi'(\alpha + 1 - s).
\]
Evaluating at $s = 0$ and using $\psi(1) = -\gamma$:
$\E[\varphi_i \mid b] = \psi(\alpha+1) + \gamma$ and
$\Var[\varphi_i \mid b] = \psi'(1) - \psi'(\alpha+1)$. The variance
bound follows from $\psi'(1) = \pi^2/6$ and $\psi' > 0$. The
unconditional mean follows from Lemma~\ref{lem:race} (the winner is $b$
with probability $p_b$) and the tower rule.
\end{proof}

\begin{corollary}[Uniform decisions]\label{cor:uniform}
If $P_i$ is uniform on $N_i$ candidates, then conditional on any
selection $\varphi_i$ is distributed as the maximum of $N_i$ i.i.d.\
$\mathrm{Exp}(1)$ variables; in particular
$\E[\varphi_i] = H_{N_i} := \sum_{j=1}^{N_i} 1/j = \ln N_i + \gamma + o(1)$
and $\Var[\varphi_i] = \sum_{j=1}^{N_i} 1/j^2 \le \pi^2/6$.
\end{corollary}

\begin{proof}
With $p_b = 1/N_i$ for all $b$, the shape is $\alpha = N_i$ and the
conditional distribution function is $(1 - e^{-t})^{N_i}$, exactly
that of the maximum of $N_i$ i.i.d.\ $\mathrm{Exp}(1)$ variables, whose
mean and variance are the classical $H_{N_i}$ and
$\sum_{j=1}^{N_i} 1/j^2$. (Equivalently: from Lemma~\ref{lem:cond} and
the recurrence $\psi(x+1) = \psi(x) + 1/x$, induction gives
$\psi(N+1) + \gamma = H_N$ and
$\psi'(1) - \psi'(N+1) = \sum_{j=1}^{N} 1/j^2$.) The conditional law
does not depend on which candidate won, so it is also the unconditional
law.
\end{proof}

\begin{corollary}[Deterministic decisions]\label{cor:pointmass}
If $P_i$ is a point mass, then $\varphi_i \sim \mathrm{Exp}(1)$,
identical to the null.
\end{corollary}

\begin{proof}
With $p_b = 1$ the shape is $\alpha = 1$ and the distribution function
is $1 - e^{-t}$, i.e., $\mathrm{Exp}(1)$. (Directly: the selection is
deterministic, so it reveals nothing about $r_{b}$, which remains
uniform.)
\end{proof}

\subsection{Proof of Theorem~\ref{thm:entropy}}\label{app:entropy}

We first isolate the trigamma estimate.

\begin{lemma}\label{lem:trigamma}
For every $y > 0$:
$\psi'(y) > \dfrac{1}{y} + \dfrac{1}{2y^2}$.
\end{lemma}

\begin{proof}
Recall $\psi'(y) = \sum_{k \ge 0} (y + k)^{-2}$. The function
$f(t) := (y + t)^{-2}$ is strictly convex on $[0, \infty)$, so the
trapezoid rule strictly overestimates each panel:
$\int_k^{k+1} f(t)\, dt < \tfrac12\bigl(f(k) + f(k+1)\bigr)$ for every
$k \ge 0$. Summing over $k$ and telescoping,
\[
\frac{1}{y} = \int_0^\infty f(t)\, dt
< \sum_{k \ge 0} f(k) - \frac{f(0)}{2}
= \psi'(y) - \frac{1}{2y^2}. \qedhere
\]
\end{proof}

\thmentropy*
\begin{proof}
By Lemma~\ref{lem:cond} it suffices to prove the scalar inequality
\begin{equation}\label{eq:scalar}
\psi(x + 1) + \gamma \;\ge\; 1 + \tfrac12 \ln x
\qquad \text{for all } x \ge 1,
\end{equation}
with equality only at $x = 1$; the theorem then follows by applying
\eqref{eq:scalar} with $x = 1/p_b \ge 1$ for each $b$ in the support of
$P_i$ and averaging:
\[
\E[\varphi_i]
= \sum_b p_b \Bigl(\psi\bigl(\tfrac1{p_b}+1\bigr) + \gamma\Bigr)
\;\ge\; \sum_b p_b \Bigl(1 + \tfrac12 \ln \tfrac{1}{p_b}\Bigr)
= 1 + \tfrac12\, \ent(P_i),
\]
with equality iff every supported $p_b$ equals $1$, i.e., iff $P_i$ is
a point mass.

To prove \eqref{eq:scalar}, set
$F(x) := \psi(x+1) + \gamma - 1 - \tfrac12 \ln x$ on $[1, \infty)$.
Since $\psi(2) = 1 - \gamma$, we have $F(1) = 0$. Differentiating and
applying Lemma~\ref{lem:trigamma} at $y = x + 1$,
\[
F'(x) = \psi'(x+1) - \frac{1}{2x}
> \frac{1}{x+1} + \frac{1}{2(x+1)^2} - \frac{1}{2x}
= \frac{2x(x+1) + x - (x+1)^2}{2x(x+1)^2}
= \frac{x^2 + x - 1}{2x(x+1)^2},
\]
which is strictly positive for $x \ge 1$. Hence $F$ is strictly
increasing on $[1, \infty)$ with $F(1) = 0$, proving \eqref{eq:scalar}
with the stated equality case.
\end{proof}

\subsection{Detection Power}\label{app:power}

\begin{corollary}[Detection power]\label{cor:power}
Suppose every effective group satisfies $\ent(P_i) \ge \bar{h} > 0$ and
all nonzero candidate probabilities lie in $[p_*, 1]$, and set
$\sigma_*^2 := \pi^2/6 + \tfrac14\bigl(\psi(1/p_* + 1) + \gamma -
1\bigr)^2$. Then for every threshold $\theta > 0$ and every
$\beta \in (0,1)$,
\[
n \;\ge\; \frac{4\bigl(\theta + \sigma_* \beta^{-1/2}\bigr)^2}{\bar{h}^2}
\quad\Longrightarrow\quad
\Pr_{\Halt}\bigl[\, z_1 \le \theta \,\bigr] \le \beta .
\]
\end{corollary}

\begin{proof}
Fix an effective group $i$ and abbreviate $\mu_i := \E[\varphi_i]$,
$g(p) := \psi(1/p + 1) + \gamma$. By Theorem~\ref{thm:entropy},
$\mu_i \ge 1 + \bar{h}/2$. By Lemma~\ref{lem:cond} and the law of total
variance,
\[
\Var[\varphi_i]
= \E_b\bigl[\Var[\varphi_i \mid b]\bigr]
+ \Var_b\bigl[g(p_b)\bigr]
\le \frac{\pi^2}{6}
+ \frac{\bigl(g(p_*) - g(1)\bigr)^2}{4}
= \sigma_*^2,
\]
where the second term uses Popoviciu's inequality \citep{popoviciu1935}
for the random variable $g(p_{b_i}) \in [g(1), g(p_*)] = [1, \psi(1/p_*+1)+\gamma]$
($g$ is decreasing and all supported probabilities lie in $[p_*, 1]$).
The scores are independent across (deduplicated) groups, so
$\E[X_1] \ge n(1 + \bar h/2)$ and $\Var[X_1] \le n \sigma_*^2$. The
event $z_1 \le \theta$ is $X_1 \le n + \theta \sqrt{n}$; whenever
$\sqrt{n}\, \bar h / 2 > \theta$, Chebyshev's inequality gives
\[
\Pr[\,X_1 \le n + \theta\sqrt{n}\,]
\le \frac{n \sigma_*^2}
{\bigl(n \bar h / 2 - \theta \sqrt{n}\bigr)^2}
= \frac{\sigma_*^2}{\bigl(\sqrt{n}\, \bar h/2 - \theta\bigr)^2}.
\]
This is at most $\beta$ as soon as
$\sqrt{n}\, \bar h/2 - \theta \ge \sigma_* / \sqrt{\beta}$, i.e.,
$n \ge 4(\theta + \sigma_* \beta^{-1/2})^2 / \bar h^2$ (which also
implies the side condition).
\end{proof}

\subsection{Deletion Robustness}\label{app:blast}

\begin{proposition}[Blast radius one: deletion attenuates but never inverts]\label{prop:blast}
Let $\tau$ be watermarked with $m$ effective single primary-observation
groups ($k_i \in \{1,2\}$: one genuine observation plus an optional
redundant record), and let $\tau'$ be obtained by deleting each
\textsc{obs} record independently with probability $r$, the
\textsc{dec} records (hence all $m$ groups) preserved. The
detector reads the genuine selected identity $b_i$, so, absent
post-deletion evaluation-point collisions (a deletion-shifted context
coinciding with another group's embedded or deletion-shifted context, a
generic condition; see the proof), every score is
stochastically at least the null, $\varphi_i \succeq \mathrm{Exp}(1)$
with $\E[\varphi_i] \ge 1$ and never below it; a group whose predecessor
retains all its records keeps its embedded above-null score,
$\E[\varphi_i] \ge 1 + \tfrac12 \ent(P_i)$. Consequently
\[
\E\bigl[X_1(\tau')\bigr]
\;\ge\; m \;+\; \tfrac12 (1-r)^2 \!\sum_{i=1}^{m} \ent(P_i)
\;\ge\; m + \tfrac12 (1-r)^2\, m\, \ent_{\min},
\]
where $\ent_{\min} := \min_i \ent(P_i)$ and all groups are effective, so
$\E[z_1(\tau')] \ge \tfrac12 (1-r)^2 \sqrt{m}\, \ent_{\min} > 0$ for
every $r < 1$.
\end{proposition}

\begin{proof}
By \eqref{eq:ctx} the context $\ctx_i$ depends only on $A_{i-1}$ and the
key, which the detector recomputes from the surviving records of group
$g_{i-1}$. Two facts drive the result. First, the \textsc{dec} records
are preserved, so the detector reads the \emph{genuine} selected
identity $b_i$ and the only effect of deletion is on the keying context:
deletion can never substitute a race loser, so, in contrast with
rewriting (Proposition~\ref{prop:loser}), it cannot push a score below
the null. Second, the window has memory one, so deleting a record from
$g_j$ disturbs the context of at most its immediate successor $g_{j+1}$
(blast radius one).

A group $g_i$ whose predecessor retains all its records has $\ctx_i$
unchanged: the detector recomputes the embedded value $r_{b_i}$, and
$\varphi_i$ keeps the above-null law of Lemma~\ref{lem:cond},
$\E[\varphi_i] \ge 1 + \tfrac12 \ent(P_i)$ by
Theorem~\ref{thm:entropy}. Otherwise the recomputed $\ctx_i$ differs
from the embedded one and the detector queries the \textsc{drbg} at the
point $(\ctx_i, b_i)$. Unless the shifted context coincides with another
group's embedded context, this point went unqueried during embedding, so
by Assumption~\ref{ass:prf} the value $r_{b_i}$ is uniform and
independent of the embedder's selection $b_i$ (made under the old
context), giving $\varphi_i \sim \mathrm{Exp}(1)$: mean exactly $1$,
never below. We take such collisions to be absent, the generic case,
since a collision requires the deletion-shortened sequence $A_{i-1}'$ to
reproduce another group's verbatim sequence $A_{j-1}$ under the injective
encoding $\mathrm{enc}$; deduplication does not remove it, because the
two share a context but not a full evaluation point $(\ctx, b)$. The same
genericity excludes two deletion-shifted groups sharing a full evaluation
point, so deduplication removes nothing and $n = m$. Were one collision
to occur, that single group would read a race loser of the collided
group and fall below the null (Proposition~\ref{prop:loser}), but never
beneath the pointwise floor $\varphi_i \ge 0$ of \eqref{eq:phi}, which
bounds every group in every case.

Group $g_{i-1}$ retains all $k_{i-1} \le 2$ of its records with
probability $(1-r)^{k_{i-1}} \ge (1-r)^2$ (group $g_1$, with the
bootstrap context, is always intact). Conditioning,
$\E[\varphi_i] \ge 1 + \tfrac12 \ent(P_i)\,(1-r)^{k_{i-1}}$, and summing
over the $m$ preserved groups gives the displayed bound.
\end{proof}

\subsection{Informed Substitution and Rewrite Invariance}\label{app:rewriteproofs}

\begin{proposition}[Informed substitution overshoots the null]\label{prop:loser}
Fix a group with $p_b < 1$ for the selected $b$. If the recorded
identity is replaced by a candidate $b' \ne b_i$ from the same group
with $p_{b'} > 0$ (so the detector evaluates a race \emph{loser}), then
$\E[-\ln(1 - r_{b'}) \mid b_i \ne b'] < 1$: the substituted score falls
strictly below the null mean. If instead $p_{b'} = 0$, the substituted
score has conditional mean exactly $1$; in no case does substitution
push the score above the null.
\end{proposition}

\begin{proof}
Unconditionally $\E[-\ln(1 - r_{b'})] = 1$ since $r_{b'}$ is uniform.
If $p_{b'} = 0$, then $b'$ is never selected (Lemma~\ref{lem:race}), so
$r_{b'}$ never affects the selection and remains an independent uniform
on the conditioning event $\{b_i \ne b'\}$, giving conditional mean
exactly $1$. Now let $p_{b'} > 0$. The
selected $b$ is realized, so $p_b > 0$ (Theorem~\ref{thm:df}), and since
$b' \neq b$ we have $p_{b'} \le 1 - p_b < 1$, so $p_{b'} \in (0,1)$. By Lemma~\ref{lem:cond},
$\E[-\ln(1 - r_{b'}) \mid b_i = b'] = \psi(1/p_{b'} + 1) + \gamma > 1$.
Since $\Pr[b_i = b'] = p_{b'}$ (Theorem~\ref{thm:df}) and
$1 = p_{b'}\,\E[-\ln(1 - r_{b'}) \mid b_i = b']
   + (1 - p_{b'})\,\E[-\ln(1 - r_{b'}) \mid b_i \neq b']$
with $1 - p_{b'} > 0$, the complementary conditional expectation is
forced strictly below $1$.
\end{proof}

\thmrewrite*
\begin{proof}
By Proposition~\ref{prop:skeleton}, $R$ preserves group positions and
the counts $(k_i)$, which are the carrier. By \eqref{eq:target}, each
target $G_i$ depends only on $(\keytwo, \texttt{qid}, i)$; the task
identifier \texttt{qid} is trajectory-level metadata reconciled against
the provider's upstream consumption record, not an editable record
content $c(e_t)$, so the keying reads nothing a rewriting attack can
alter. Every indicator in \eqref{eq:z2} is
therefore unchanged. Finally, the pooled index set is unchanged as
well: by Definition~\ref{def:attrs}, effectiveness is evaluated against
the environment-supplied $B_i$ and the executed action stream, neither
of which $R$ can edit, and the tally evaluation points
$H(\texttt{qid} \,\|\, i)$ are content-free, so deduplication reads
nothing $R$ can touch. Hence the summation set and its size $n$
coincide for $\tau$ and $R(\tau)$, giving $X_2(R(\tau)) = X_2(\tau)$
and $z_2(R(\tau)) = z_2(\tau)$.
\end{proof}

\subsection{Proof of Theorem~\ref{thm:erasure}}\label{app:erasure}

\thmerasure*
\begin{proof}
\textbf{(a)} If $s(\tau') = s(\tau)$, Theorem~\ref{thm:rewrite} gives
$z_2(\tau') = z_2(\tau)$, and Lemma~\ref{lem:null2} together with the
embedding rule gives $z_2(\tau) = \sqrt{m}$ for the watermarked $\tau$
with admissible augmentation. The contrapositive is the second
sentence.

\textbf{(b)} The preserved decisions keep all $m$ groups, so $n = m$
groups are scored. Call a group \emph{intact} if its selection-channel
evaluation point is uncorrupted; by hypothesis at least $m - a$ are
intact. Because the corrupted set is chosen obliviously, it is
independent of the realized race values, so conditioning on it leaves
each intact group's score with conditional mean at least
$1 + \tfrac12 \ent(P_i) \ge 1 + \tfrac{h}{2}$
(Theorem~\ref{thm:entropy} and the entropy floor $h$); summing and
taking expectations,
\[
\E\Bigl[\textstyle\sum_{i\, \mathrm{intact}} \varphi_i\Bigr]
\;\ge\; \Bigl(1 + \frac{h}{2}\Bigr)\,
\E\bigl[\#\{i \text{ intact}\}\bigr]
\;\ge\; (m - a)\Bigl(1 + \frac{h}{2}\Bigr),
\]
which needs only that the kept set is independent of the realized
scores, not of the elicited entropies. (Without
obliviousness this step fails: an attacker targeting groups whose
realized winner was improbable removes more than $1 + h/2$ per
corrupted group; see Remark~\ref{rem:adaptive}.) For the corrupted
groups we use only $\varphi_i \ge 0$, which holds pointwise by
\eqref{eq:phi}: a deletion-shifted context contributes mean $1$
(Proposition~\ref{prop:blast}) and a rewriting substitution mean below
$1$ (Proposition~\ref{prop:loser}), both nonnegative; note we may
\emph{not} credit the rewritten groups with the null mean $1$.
Therefore, with $n = m$,
\[
\E\bigl[X_1(\tau')\bigr] \ge (m - a)\Bigl(1 + \frac{h}{2}\Bigr),
\qquad
\E\bigl[z_1(\tau')\bigr]
\ge \frac{(m - a)(1 + h/2) - m}{\sqrt{m}} .
\]
Now suppose $\E[z_1(\tau')] \le \theta$ for a threshold $\theta \ge 0$.
Rearranging,
\[
(m - a)\Bigl(1 + \frac{h}{2}\Bigr)
\;\le\; m + \theta \sqrt{m},
\]
hence
\[
a
\;\ge\; m - \frac{m + \theta \sqrt{m}}{1 + h/2}
\;=\; \frac{h}{2 + h}\, m - \frac{\theta \sqrt{m}}{1 + h/2}
\;\ge\; \frac{h}{2 + h}\, m - \theta \sqrt{m}. \qedhere
\]
\end{proof}

\begin{remark}
A single content edit can alter at most two groups' selection-channel evaluation
points (its own selected identity and, through the memory-one
window, its successor's context), so the bound on $a$ translates into
a bound of at least $a/2$ on the number of \emph{edited records}; the
constant-fraction conclusion is unchanged.
\end{remark}

\begin{remark}[Why obliviousness is needed in (b)]\label{rem:adaptive}
The hypothesis is necessary, not an artifact of the proof. By
Lemma~\ref{lem:cond} the conditional mean of a score given its winner,
$\psi(1/p_b + 1) + \gamma$, is unbounded as $p_b \to 0$ while
$\ent(P_i)$ can stay small, so an attacker that observes the realized
selections and targets precisely the groups whose winner was improbable
removes more than $1 + h/2$ of expected signal per altered group, and
with sufficiently skewed decision distributions it suppresses $z_1$ on
a budget below the bound. The oblivious class still covers intervention
sets chosen by position, at random, or by any rule ignorant of the
realized selections; in particular it covers the attacks of
Section~\ref{sec:exp}, where the LLM rewriter attacks a random fraction
$q$ of groups while remaining fully informed \emph{within} each
attacked group. Quantifying the score-adaptive rate, which involves the
order statistics of the conditional laws of Lemma~\ref{lem:cond}, is
left open.
\end{remark}

\subsection{Composition of the Two Tests}\label{app:compose}

\begin{proposition}[One-way coupling and joint FPR]\label{prop:compose}
The tally-channel targets $(G_i)$ are independent of $\keyone$ and of all
selection-channel randomness; under $\Hnull$ the exact $p$-values $p_1$ and
$p_2$ are independent (the $z$-scores, sharing the random count $n$, are
independent only conditional on the trajectory). The combined test that
rejects when either layer's exact $p$-value falls below $\alpha/2$ has
false positive rate at most $\alpha$ by the union bound; independence of
$(p_1, p_2)$ further licenses the Fisher and Stouffer combinations, which
remain valid but are conservative because the Binomial $p_2$ is discrete.
\end{proposition}

\begin{proof}
By \eqref{eq:target} the targets $(G_i)$ are functions of
$(\keytwo, \texttt{qid}, i)$ only. Under Assumption~\ref{ass:prf},
DRBG calls under the independent keys $\keyone$ and $\keytwo$ are
mutually independent, so $(G_i)$ is independent of the family $(r_b)$
of \eqref{eq:rb} and of everything computed from it. Under $\Hnull$ the
trajectory $\tau$ is independent of both keys, and $\keyone \perp
\keytwo$. Condition on the pair $(\tau, X_2)$, which fixes the skeleton,
the counts $(k_i)$, every selection-channel evaluation point, the
deduplicated index set, the count $n$, and the tally statistic. Given
$(\tau, X_2)$, the selection statistic $X_1$ is a function of the
$\keyone$-values at the now-fixed evaluation points alone, and, since
$\keyone$ is independent of $\tau$ and of $\keytwo$ (hence of $X_2$), it
remains $\mathrm{Gamma}(n,1)$; so by Lemma~\ref{lem:null1} the exact
$p$-value $p_1 = Q(n, X_1)$ is uniform on $(0,1)$ given $(\tau, X_2)$ (the
probability integral transform). As this holds for every value of
$(\tau, X_2)$, $p_1$ is independent of $(\tau, X_2)$, and in particular of
$p_2$, which is a function of $(n, X_2)$. The $z$-scores are
not in general independent marginally, as both standardize by the same
random $n = n(\tau)$; the dependence vanishes only at the $p$-value
level. The union bound,
which uses no independence, gives
$\Pr_{\Hnull}[\,p_1 \le \alpha/2 \text{ or } p_2 \le \alpha/2\,] \le
\alpha$. Independence of $(p_1, p_2)$ licenses Fisher's method
\citep{Fisher1927StatisticalMF} applied to $-2\ln p_1 - 2\ln p_2$ and Stouffer's to
$(z_1 + z_2)/\sqrt{2}$; because the Binomial $p_2$ is discrete
(super-uniform), the $\chi^2_4$ and normal references are conservative
rather than exact.
\end{proof}

\section{Practical Notes}\label{app:practical}

\paragraph{Elicitation of $P_i$.}
The distribution $P_i$ is elicited from the agent at decision time
(e.g., as normalized scores over $B_i$); when parsing fails we fall
back to the uniform distribution on $B_i$, which by
Corollary~\ref{cor:uniform} is the most detectable case and by
Theorem~\ref{thm:df} still leaves the realized selection a faithful
sample of the fallback distribution actually used.

\paragraph{Estimating $p_0$.}
The tally-channel null rate $p_0 = \tfrac12$ is exact whenever the baseline
agent's group counts satisfy $k_i \equiv 1$,
because the balanced keyed pattern \eqref{eq:target} then makes each
hit a fair coin regardless of the baseline's behavior. A baseline that
emits $k_i = 2$ keeps a fair-coin hit, but a group with
$k_i \notin \{1,2\}$ is a guaranteed \emph{miss}, since neither
$\{1\}$ nor $\{2\}$ contains $k_i$. Such groups only lower the null
hit rate below $\tfrac12$, so for heavy-observation baselines the
$\mathrm{Bin}(n,\tfrac12)$ tail is a \emph{conservative} bound on the
null $p$-value rather than exact; we report empirically calibrated
$p_0$ for transparency.

\paragraph{Log/execution consistency audit.}
Operationally, the verifier compares the skeleton of the
reseller-released log against the grouping reconstructed from the
provider-side execution record; any mismatch in \textsc{dec} counts or
group boundaries flags tampering (Section~\ref{sec:threat}). The audit
is assumed available rather than evaluated in our experiments.

\paragraph{Admissibility.}
When a group's logging format does not admit a redundant record, the
embedder leaves $k_i = 1$; if the target was $\{2\}$ this group becomes
a guaranteed miss. With admissibility rate $\eta$ over target-$2$
groups, the watermarked hit rate is $1 - \tfrac12(1 - \eta)$ and the
detection statistics adjust in the obvious way; all groups remain
usable.

\begin{remark}[Window knob]\label{rem:window}
Whether the selection channel's window \emph{includes} redundant records is
a design choice. Including them (our default) couples the layers: a
targeted deletion of redundant records then also corrupts the
selection-channel contexts of the following groups, so the selection channel
shares a sliver of the tally channel's deletion exposure, bounded as in
Proposition~\ref{prop:blast}. Sliding the window over genuine actions
only makes the selection channel entirely immune to targeted
redundant-record deletion, at the price of decoupling; our experiments
use the coupled default (Section~\ref{sec:exp}). A window longer than
one is equally admissible but widens the damage of each deletion from
one group to the window length.
\end{remark}

\paragraph{Numerical verification.}
All distributional claims in this paper (Lemma~\ref{lem:race}
(independence and marginals of the race), Theorem~\ref{thm:df},
Lemma~\ref{lem:cond} (conditional CDF, digamma mean, trigamma
variance), Corollaries~\ref{cor:uniform}--\ref{cor:pointmass},
Theorem~\ref{thm:entropy} (including Lemma~\ref{lem:trigamma} on a
dense grid), Lemma~\ref{lem:null1}, Proposition~\ref{prop:loser},
Corollary~\ref{cor:power}, Proposition~\ref{prop:blast},
Lemma~\ref{lem:null2}, and Proposition~\ref{prop:compose}) were
additionally verified by Monte Carlo simulation; the script is included
in the supplementary material.

\section{Additional Experimental Results}\label{app:exp}

This appendix records the per-split detection values, per-seed raw
values, the backbone ablation, the full baseline-robustness tables,
and the calibrated-FPR detection-power tables behind
Section~\ref{sec:exp}. Table~\ref{tab:splits} maps the split labels
used throughout to the benchmarks' native identifiers.

\begin{table}[!ht]
\centering
\caption{Split labels and the benchmarks' native identifiers.}
\label{tab:splits}
\small
\begin{tabular}{cl|cl}
\toprule
\textbf{Label} & \textbf{ToolBench split} & \textbf{Label} & \textbf{ALFWorld task type} \\
\midrule\midrule
T1 & \texttt{G1\_category}    & A1 & \texttt{pick\_and\_place} \\
T2 & \texttt{G1\_instruction} & A2 & \texttt{pick\_clean\_then\_place} \\
T3 & \texttt{G1\_tool}        & A3 & \texttt{pick\_heat\_then\_place} \\
T4 & \texttt{G2\_category}    & A4 & \texttt{pick\_cool\_then\_place} \\
T5 & \texttt{G2\_instruction} & A5 & \texttt{look\_at\_obj\_in\_light} \\
T6 & \texttt{G3\_instruction} & A6 & \texttt{pick\_two\_obj\_and\_place} \\
\midrule
\multicolumn{2}{l|}{ID = \texttt{valid\_seen}} &
\multicolumn{2}{l}{OOD = \texttt{valid\_unseen}} \\
\bottomrule
\end{tabular}
\end{table} All values
are mean $\pm$ sample standard deviation over three seeds unless
marked. The per-task utility breakdown appears in
Table~\ref{tab:main}; further per-seed raw values are included in the
supplementary material.

\subsection{ToolBench Detection per Split}

\begin{table}[!ht]
\centering
\caption{ToolBench detection per split: RG single-layer $z$ and
\method{} per-channel $z$, each with its wrong-key control.}
\label{tab:tb-detect}
\footnotesize
\setlength{\tabcolsep}{3.4pt}
\begin{tabular}{l|cc|cccc}
\toprule
\multirow{2}{*}{\textbf{Split}} & \multicolumn{2}{c|}{\textbf{RG}} & \multicolumn{4}{c}{\textbf{\method{}}} \\
\cmidrule(lr){2-3} \cmidrule(lr){4-7}
& $z$ & wk & Sel.\ $z$ & Tally $z$ & wk sel. & wk tally \\
\midrule\midrule
T1 & $0.65{\pm}0.92$ & $0.62{\pm}0.89$ & $3.78{\pm}0.23$ & $5.12{\pm}0.40$ & $-0.30{\pm}0.41$ & $0.19{\pm}0.40$ \\
T2 & $1.87{\pm}0.34$ & $-0.46{\pm}0.61$ & $5.40{\pm}2.16$ & $5.32{\pm}0.29$ & $1.07{\pm}0.95$ & $1.29{\pm}0.69$ \\
T3 & $3.04{\pm}0.23$ & $-1.77{\pm}0.96$ & $6.03{\pm}0.44$ & $5.44{\pm}0.37$ & $-0.10{\pm}0.54$ & $1.42{\pm}0.55$ \\
T4 & $4.32{\pm}0.81$ & $-1.70{\pm}1.08$ & $2.10{\pm}2.00$ & $5.54{\pm}0.70$ & $-0.35{\pm}0.25$ & $-0.44{\pm}0.51$ \\
T5 & $4.23{\pm}0.18$ & $-0.04{\pm}0.02$ & $6.72{\pm}1.67$ & $7.77{\pm}0.04$ & $-1.66{\pm}0.48$ & $0.64{\pm}0.23$ \\
T6 & $2.23{\pm}0.53$ & $-1.94{\pm}0.66$ & $3.04{\pm}0.57$ & $5.41{\pm}1.06$ & $0.37{\pm}0.36$ & $-0.54{\pm}0.37$ \\
\midrule
\textbf{Mean over splits} & $2.72$ & $-0.88$ & $4.51$ & $5.77$ & $-0.16$ & $0.43$ \\
\bottomrule
\end{tabular}
\end{table}

\FloatBarrier
\subsection{ALFWorld per Seed}

\begin{table}[!ht]
\centering
\caption{\method{} ALFWorld per-seed values behind the headline
aggregates.}
\label{tab:alf-seed}
\small
\begin{tabular}{l|ccc|ccc}
\toprule
\multirow{2}{*}{\textbf{Split}} & \multicolumn{3}{c|}{\textbf{SR (\%)}} & \multicolumn{3}{c}{\textbf{Sel.\ $z$}} \\
\cmidrule(lr){2-4} \cmidrule(lr){5-7}
& seed 1 & seed 2 & seed 3 & seed 1 & seed 2 & seed 3 \\
\midrule\midrule
ID (\texttt{valid\_seen}) & $85.7$ & $80.7$ & $84.3$ & $87.5$ & $99.6$ & $95.3$ \\
OOD (\texttt{valid\_unseen}) & $85.8$ & $79.9$ & $77.6$ & $95.4$ & $106.4$ & $105.8$ \\
\bottomrule
\end{tabular}
\end{table}

\FloatBarrier
\subsection{Redundant-Record Accounting}\label{app:redundant}

Table~\ref{tab:redundant} itemizes the redundant records behind the
\method{} step counts of Table~\ref{tab:main}: per task, the decision
steps the agent actually takes, the redundant records the tally
channel appends, and their total, for every benchmark split and both
backbones. The volume scales with trajectory length, about half a
record per decision group on ToolBench and eleven to fifteen records
per task on ALFWorld horizons, and the records disturb neither the
semantic content of the log nor the task's execution
(Definition~\ref{def:redundant}): the decision-step column tracks
\textsc{Base} throughout, so the accounting confirms that the entire
step overhead in Table~\ref{tab:main} is watermark records rather than
agent work.

\begin{table}[!ht]
\centering
\caption{Redundant-record accounting for \method{} per task: decision
steps, appended redundant records, and total logged steps
($\dagger$: difference of the two measured columns).}
\label{tab:redundant}
\small
\begin{tabular}{l|ccc}
\toprule
\textbf{Setting} & \textbf{Decision steps} & \textbf{Redundant records} & \textbf{Total} \\
\midrule\midrule
ToolBench T1 & $1.32{\pm}0.21$ & $0.48^{\dagger}$ & $1.80{\pm}0.22$ \\
ToolBench T2 & $1.42{\pm}0.15$ & $0.51^{\dagger}$ & $1.93{\pm}0.13$ \\
ToolBench T3 & $1.48{\pm}0.20$ & $0.89^{\dagger}$ & $2.37{\pm}0.34$ \\
ToolBench T4 & $1.55{\pm}0.40$ & $0.85^{\dagger}$ & $2.40{\pm}0.66$ \\
ToolBench T5 & $1.21{\pm}0.01$ & $0.60^{\dagger}$ & $1.81{\pm}0.05$ \\
ToolBench T6 & $1.50{\pm}0.59$ & $0.85^{\dagger}$ & $2.35{\pm}0.90$ \\
\textbf{ToolBench Avg.} & $1.37{\pm}0.09$ & $0.68^{\dagger}$ & $2.05{\pm}0.16$ \\
\midrule
ALFWorld ID & $23.15{\pm}1.73$ & $11.04{\pm}0.85$ & $34.19{\pm}2.58$ \\
ALFWorld OOD & $24.80{\pm}1.42$ & $11.85{\pm}0.71$ & $36.65{\pm}2.13$ \\
\midrule
ALFWorld ID (Qwen) & $30.28{\pm}0.17$ & $14.68{\pm}0.15$ & $44.96{\pm}0.28$ \\
ALFWorld OOD (Qwen) & $29.15{\pm}0.24$ & $13.82{\pm}0.15$ & $42.97{\pm}0.39$ \\
\bottomrule
\end{tabular}
\end{table}

\FloatBarrier
\subsection{Backbone Ablation (Local Qwen under vLLM)}\label{app:qwen}

The second backbone is a locally deployed Qwen3-4B-Instruct model served with
vLLM. The ablation runs ALFWorld only: a ToolBench decision presents
dozens of candidate tools at once, more than a 4B-parameter backbone
can reliably discriminate among, so that benchmark is uninformative at
this model scale. On ALFWorld the baselines transfer: RG reaches
$z = 38.16 \pm 0.98$ (ID) and $35.63 \pm 0.75$ (OOD) with wrong-key
controls near zero, AM-F recovers $56.05$ and $54.98$ bits per task,
and the qualitative utility picture matches the main backbone, RG
paying the largest cost while the distribution-preserving arms stay
near \textsc{Base}
(Tables~\ref{tab:qwen-util} and \ref{tab:qwen-detect}). \method{}
itself transfers as well, reaching SR $63.8 \pm 0.4$ (ID) and
$67.2 \pm 1.3$ (OOD), $3.1$ and $4.3$ points above RG on the same
backbone and task subsets: the distortion-free advantage is even more
visible at the 4B scale, where the model tolerates less interference
with its sampling. Its step counts again include the tally channel's
redundant records, $14.7$ (ID) and $13.8$ (OOD) per task at this
backbone's lower success rate, where failed episodes run to the
$50$-step cap; net of them the decision path runs $30.3$ and $29.2$
steps against \textsc{Base}'s $30.5$ and $30.8$, at parity here as on
the main backbone. Detection reaches selection $z = 90.21$ (ID) and
$93.88$ (OOD) with the tally at $62.96$ and $60.44$, and all wrong-key
controls stay far from the threshold.

\begin{table}[!ht]
\centering
\caption{Qwen backbone, ALFWorld utility ($\dagger$: $n$-weighted mean
over task types; \method{} steps include the tally channel's redundant
records).}
\label{tab:qwen-util}
\small
\begin{tabular}{l|cc|cc}
\toprule
\multirow{2}{*}{\textbf{Arm}} & \multicolumn{2}{c|}{\textbf{SR (\%)} $\uparrow$}
& \multicolumn{2}{c}{\textbf{Steps / task}} \\
\cmidrule(lr){2-3} \cmidrule(lr){4-5}
& ID & OOD & ID & OOD \\
\midrule\midrule
Base & $64.5{\pm}0.8$ & $62.9{\pm}2.3$ & $30.5^{\dagger}$ & $30.8^{\dagger}$ \\
AM-F & $67.4{\pm}0.8$ & $66.9{\pm}5.2$ & $29.7^{\dagger}$ & $29.5^{\dagger}$ \\
RG & $60.7{\pm}3.8$ & $62.9{\pm}0.4$ & $31.7^{\dagger}$ & $31.3^{\dagger}$ \\
\method{} & $63.8{\pm}0.4$ & $67.2{\pm}1.3$ & $45.0{\pm}0.3$ & $43.0{\pm}0.4$ \\
\bottomrule
\end{tabular}
\end{table}

\begin{table}[!ht]
\centering
\caption{Qwen backbone, ALFWorld detection.}
\label{tab:qwen-detect}
\small
\begin{tabular}{l|cc|cc|c}
\toprule
\textbf{Setting} & \textbf{RG $z$} $\uparrow$ & \textbf{RG wk}
& \textbf{\method{} Sel.\ $z$} $\uparrow$
& \textbf{\method{} Tally $z$} $\uparrow$
& \textbf{AM-F bits/task} $\uparrow$ \\
\midrule\midrule
ALFWorld ID & $38.16{\pm}0.98$ & $0.15{\pm}0.35$ & $90.21{\pm}2.00$ & $62.96{\pm}0.17$ & $56.05{\pm}1.93$ \\
ALFWorld OOD & $35.63{\pm}0.75$ & $0.46{\pm}0.91$ & $93.88{\pm}2.05$ & $60.44{\pm}0.31$ & $54.98{\pm}4.82$ \\
\bottomrule
\end{tabular}
\end{table}

\FloatBarrier
\subsection{Baseline Robustness in Full}

\begin{table}[!ht]
\centering
\caption{RG and AM-F under the deletion sweep on ToolBench ($15$ runs
per rate).}
\label{tab:base-del}
\small
\begin{tabular}{l|cccc}
\toprule
\textbf{Deletion rate} & \textbf{RG $z$} & \textbf{AM-F $z$} & \textbf{AM-F bit acc.} & \textbf{AM-F RLNC succ.} \\
\midrule\midrule
$0.0$ & $7.09$ & $15.85$ & $1.00$ & $0.56$ \\
$0.1$ & $6.58$ & $14.98$ & $1.00$ & $0.54$ \\
$0.2$ & $5.61$ & $14.12$ & $1.00$ & $0.56$ \\
$0.3$ & $5.07$ & $13.15$ & $1.00$ & $0.52$ \\
$0.4$ & $4.74$ & $11.97$ & $1.00$ & $0.47$ \\
$0.5$ & $4.27$ & $11.26$ & $1.00$ & $0.49$ \\
$0.6$ & $3.20$ & $10.10$ & $1.00$ & $0.46$ \\
$0.7$ & $2.80$ & $8.67$ & $1.00$ & $0.24$ \\
\bottomrule
\end{tabular}
\end{table}

\begin{table}[!ht]
\centering
\caption{RG and AM-F under the rewriting sweep on ToolBench ($9$ runs
per strength).}
\label{tab:base-sub}
\small
\begin{tabular}{l|cccc}
\toprule
\textbf{Substitution $q$} & \textbf{RG $z$} & \textbf{AM-F $z$} & \textbf{AM-F bit acc.} & \textbf{AM-F RLNC succ.} \\
\midrule\midrule
$0.0$ & $7.09$ & $15.85$ & $1.00$ & $0.56$ \\
$0.2$ & $5.39$ & $13.99$ & $0.95$ & $0.44$ \\
$0.4$ & $3.79$ & $11.56$ & $0.88$ & $0.20$ \\
$0.6$ & $1.33$ & $1.47$ & $0.55$ & $0.04$ \\
$0.8$ & $0.60$ & $-1.62$ & $0.44$ & $0.00$ \\
$1.0$ & $0.29$ & $-4.04$ & $0.34$ & $0.00$ \\
\bottomrule
\end{tabular}
\end{table}

\begin{table}[!ht]
\centering
\caption{RG and AM-F under the combined attack on ToolBench: detection
$z$ at deletion rate $r$ and rewriting strength $q$ (nine runs per
cell).}
\label{tab:base-comb}
\small
\setlength{\tabcolsep}{4.6pt}
\begin{tabular}{l|rrrrrr}
\toprule
\textbf{Arm, rate} & $\bm{q{=}0.0}$ & $\bm{q{=}0.2}$ & $\bm{q{=}0.4}$ & $\bm{q{=}0.6}$ & $\bm{q{=}0.8}$ & $\bm{q{=}1.0}$ \\
\midrule\midrule
RG, $r = 0.3$ & $5.16$ & $4.38$ & $3.38$ & $1.91$ & $1.68$ & $1.37$ \\
RG, $r = 0.5$ & $4.29$ & $3.37$ & $2.70$ & $2.32$ & $1.45$ & $0.41$ \\
RG, $r = 0.7$ & $2.93$ & $2.35$ & $2.15$ & $0.62$ & $0.74$ & $0.48$ \\
\midrule
AM-F, $r = 0.3$ & $13.07$ & $11.90$ & $10.02$ & $1.70$ & $-0.54$ & $-3.12$ \\
AM-F, $r = 0.5$ & $11.27$ & $9.98$ & $8.10$ & $1.72$ & $-0.78$ & $-2.69$ \\
AM-F, $r = 0.7$ & $6.05$ & $6.06$ & $5.56$ & $-0.74$ & $-1.25$ & $-2.11$ \\
\bottomrule
\end{tabular}
\end{table}

\FloatBarrier
\subsection{Detection Power at Calibrated FPR}\label{app:tpr}

\paragraph{Protocol.}
Positives are watermarked trajectories scored under the true key;
negatives are the same trajectories scored under the wrong key
($99990001$/$99990002$). For clean detection, $B$ trajectories are
pooled into one bundle statistic
$z_{\mathrm{bundle}} = (\sum X - \sum E)/\sqrt{\sum V}$, and
TPR@$x\%$FPR is the fraction of positive bundles exceeding the
$(1{-}x\%)$ quantile of the negative bundles ($5000$ bootstrap
resamples). For detection under attack, each attack cell already pools
$118$ to $127$ trajectories into one $z$, so one cell is one
fixed-size bundle; TPR@$1\%$FPR is the fraction of cells above
$\theta_{1\%} = 2.326$, the one-sided Gaussian $1\%$ threshold, which
the wrong-key per-cell $z$ calibrates (granularity $1/15$ per deletion
strength, $1/9$ otherwise). Grand-pooled wrong-key selection scores
sit slightly below zero (about $-0.06$ to $-0.10$ per group), the
harmless side of a one-sided upper-tail test; TPR thresholds are taken
from the null's own quantiles, so the shift cannot inflate false
positives.

\begin{table}[!ht]
\centering
\caption{Clean detection power: single-trajectory AUC and TPR at $1\%$
FPR for bundle sizes $B \in \{1, 10, 50\}$.}
\label{tab:tpr-clean}
\small
\begin{tabular}{ll|c|ccc}
\toprule
\textbf{Setting} & \textbf{Scheme / channel} & \textbf{AUC} ($B{=}1$)
& $B{=}1$ & $B{=}10$ & $B{=}50$ \\
\midrule\midrule
\multirow{4}{*}{ToolBench}
 & \method{} selection & $0.711$ & $0.080$ & $0.618$ & $1.000$ \\
 & \method{} tally     & $0.802$ & $0.083$ & $1.000$ & $1.000$ \\
 & RG                  & $0.650$ & $0.030$ & $0.563$ & $1.000$ \\
 & AM-F                & $0.882$ & $0.130$ & $1.000$ & $1.000$ \\
\midrule
\multirow{3}{*}{ALFWorld ID}
 & \method{} selection & $0.974$ & $0.874$ & $1.000$ & $1.000$ \\
 & \method{} tally     & $0.985$ & $0.850$ & $1.000$ & $1.000$ \\
 & RG                  & $0.933$ & $0.562$ & $1.000$ & $1.000$ \\
\midrule
\multirow{3}{*}{ALFWorld OOD}
 & \method{} selection & $0.995$ & $0.938$ & $1.000$ & $1.000$ \\
 & \method{} tally     & $0.995$ & $0.860$ & $1.000$ & $1.000$ \\
 & RG                  & $0.900$ & $0.563$ & $1.000$ & $1.000$ \\
\bottomrule
\end{tabular}
\end{table}

\begin{table}[!ht]
\centering
\caption{Single-trajectory ($B = 1$) TPR at three FPR levels.}
\label{tab:tpr-fpr}
\small
\begin{tabular}{ll|ccc}
\toprule
\textbf{Setting} & \textbf{Scheme / channel}
& \textbf{@5\%} & \textbf{@1\%} & \textbf{@0.1\%} \\
\midrule\midrule
\multirow{4}{*}{ToolBench}
 & \method{} selection & $0.260$ & $0.080$ & $0.051$ \\
 & \method{} tally     & $0.147$ & $0.083$ & $0.033$ \\
 & RG                  & $0.075$ & $0.030$ & $0.030$ \\
 & AM-F                & $0.227$ & $0.130$ & $0.082$ \\
\midrule
\multirow{3}{*}{ALFWorld ID}
 & \method{} selection & $0.924$ & $0.874$ & $0.648$ \\
 & \method{} tally     & $0.930$ & $0.850$ & $0.786$ \\
 & RG                  & $0.721$ & $0.562$ & $0.351$ \\
\midrule
\multirow{3}{*}{ALFWorld OOD}
 & \method{} selection & $0.979$ & $0.938$ & $0.865$ \\
 & \method{} tally     & $0.959$ & $0.860$ & $0.826$ \\
 & RG                  & $0.677$ & $0.563$ & $0.457$ \\
\bottomrule
\end{tabular}
\end{table}

\begin{table}[!ht]
\centering
\caption{TPR at $1\%$ FPR under single-axis attacks (ToolBench;
bundle-level, see protocol).}
\label{tab:tpr-robust}
\small
\setlength{\tabcolsep}{5pt}
\begin{tabular}{l|cccc||l|cccc}
\toprule
\textbf{$r$} & \textbf{Sel.} & \textbf{Tally} & \textbf{RG} & \textbf{AM-F}
& \textbf{$q$} & \textbf{Sel.} & \textbf{Tally} & \textbf{RG} & \textbf{AM-F} \\
\midrule\midrule
$0.0$ & $1.000$ & $1.000$ & $1.000$ & $1.000$ & $0.0$ & $1.000$ & $1.000$ & $1.000$ & $1.000$ \\
$0.1$ & $1.000$ & $1.000$ & $1.000$ & $1.000$ & $0.2$ & $1.000$ & $1.000$ & $1.000$ & $1.000$ \\
$0.2$ & $1.000$ & $1.000$ & $1.000$ & $1.000$ & $0.4$ & $1.000$ & $1.000$ & $1.000$ & $1.000$ \\
$0.3$ & $1.000$ & $1.000$ & $1.000$ & $1.000$ & $0.6$ & $0.000$ & $1.000$ & $0.333$ & $0.111$ \\
$0.4$ & $1.000$ & $0.867$ & $1.000$ & $1.000$ & $0.8$ & $0.000$ & $1.000$ & $0.000$ & $0.000$ \\
$0.5$ & $1.000$ & $0.267$ & $1.000$ & $1.000$ & $1.0$ & $0.000$ & $1.000$ & $0.000$ & $0.000$ \\
$0.6$ & $1.000$ & $0.000$ & $0.933$ & $1.000$ & & & & & \\
$0.7$ & $0.933$ & $0.000$ & $0.733$ & $1.000$ & & & & & \\
\bottomrule
\end{tabular}
\end{table}

\begin{table}[!ht]
\centering
\caption{TPR at $1\%$ FPR under the combined attack (ToolBench; nine
cells per $(r, q)$).}
\label{tab:tpr-comb}
\small
\setlength{\tabcolsep}{5pt}
\begin{tabular}{cc|cccc}
\toprule
\textbf{$r$} & \textbf{$q$} & \textbf{\method{} Sel.} & \textbf{\method{} Tally}
& \textbf{RG} & \textbf{AM-F} \\
\midrule\midrule
$0.3$ & $0.0$ & $1.000$ & $1.000$ & $1.000$ & $1.000$ \\
$0.3$ & $0.2$ & $1.000$ & $1.000$ & $1.000$ & $1.000$ \\
$0.3$ & $0.4$ & $0.889$ & $1.000$ & $0.889$ & $1.000$ \\
$0.3$ & $0.6$ & $0.000$ & $1.000$ & $0.333$ & $0.222$ \\
$0.3$ & $0.8$ & $0.000$ & $1.000$ & $0.111$ & $0.000$ \\
$0.3$ & $1.0$ & $0.000$ & $1.000$ & $0.000$ & $0.000$ \\
\midrule
$0.5$ & $0.0$ & $1.000$ & $0.444$ & $1.000$ & $1.000$ \\
$0.5$ & $0.2$ & $1.000$ & $0.444$ & $0.889$ & $1.000$ \\
$0.5$ & $0.4$ & $0.778$ & $0.444$ & $0.778$ & $1.000$ \\
$0.5$ & $0.6$ & $0.000$ & $0.667$ & $0.556$ & $0.222$ \\
$0.5$ & $0.8$ & $0.000$ & $0.222$ & $0.111$ & $0.000$ \\
$0.5$ & $1.0$ & $0.000$ & $0.333$ & $0.000$ & $0.000$ \\
\midrule
$0.7$ & $0.0$ & $1.000$ & $0.000$ & $0.778$ & $1.000$ \\
$0.7$ & $0.2$ & $0.889$ & $0.000$ & $0.556$ & $1.000$ \\
$0.7$ & $0.4$ & $0.667$ & $0.000$ & $0.444$ & $1.000$ \\
$0.7$ & $0.6$ & $0.000$ & $0.000$ & $0.111$ & $0.000$ \\
$0.7$ & $0.8$ & $0.000$ & $0.000$ & $0.000$ & $0.000$ \\
$0.7$ & $1.0$ & $0.000$ & $0.000$ & $0.111$ & $0.000$ \\
\bottomrule
\end{tabular}
\end{table}

\end{document}